\documentclass[11pt]{article}

\usepackage{fullpage}
\usepackage{amssymb,amsfonts,xspace,graphicx,relsize,bm,mathtools,xcolor,amsmath,breqn,algorithm,algpseudocode,mathtools,mathrsfs,multirow}
\usepackage[T1]{fontenc}
\usepackage[utf8]{inputenc}
\usepackage[pagebackref]{hyperref}

\usepackage[thmmarks]{ntheorem}

\usepackage{cleveref}
\usepackage{thm-restate}
\usepackage{array}
\usepackage{parskip}

\usepackage[usestackEOL]{stackengine}

\usepackage{url}
\def\01{\{0,1\}}

\newcommand{\eps}{\varepsilon}
\newcommand{\ket}[1]{|#1\rangle}

\newcommand{\braket}[2]{\langle#1|#2\rangle}
\newcommand{\Tr}{\mbox{\rm Tr}}

\newcommand{\norm}[1]{\mbox{$\parallel{#1}\parallel$}}
\newcommand{\Id}{I}

\newcommand{\Cc}{{\mathcal C}} 

\newcommand{\X}{\ensuremath{\mathcal{X}}}

\newcommand{\A}{\ensuremath{\mathcal{A}}}

\newcommand{\R}{\ensuremath{\mathbb{R}}}

\newcommand{\G}{\ensuremath{\mathcal{G}}}

\usepackage{ntheorem}

\DeclareMathOperator{\poly}{poly}

\DeclareMathOperator{\polylog}{polylog}

\newcommand{\Se}{\ensuremath{\mathcal{S}}}

\newtheorem{theorem}{Theorem}[section]
\newtheorem{definition}[theorem]{Definition}

\newtheorem{lemma}[theorem]{Lemma}

\newtheorem{claim}[theorem]{Claim}
\newcommand{\pmset}[1]{\{-1,1\}^{#1}} 

\newtheorem{innercustomthm}{Claim}
\newenvironment{customthm}[1]
  {\renewcommand\theinnercustomthm{#1}\innercustomthm}
  {\endinnercustomthm}

\def\01{\{0,1\}}

\newcommand{\sgn}{\mathrm{sign}}
\DeclareMathOperator{\sign}{sign}

\usepackage{caption}

\newcommand{\VC}{\ensuremath{\mathsf{VC}}}

\usepackage{array}
\newcolumntype{M}[1]{>{\centering\arraybackslash}m{#1}}

\usepackage{kpfonts}

\usepackage[margin=1in]{geometry}
\hypersetup{
	colorlinks,
	linkcolor={blue!100!black},
	citecolor={blue!100!black},
}
\usepackage{comment}
\usepackage{multicol}

\newenvironment{proof}
{\noindent {\bf Proof. }}
{{\hfill $\Box$}\\
	\smallskip}

\algtext*{EndFor}

\renewcommand\thmcontinues[1]{Formal Statement}

\usepackage{footnote}
\makesavenoteenv{tabular}

\makeatletter
\newcommand{\algmargin}{\the\ALG@thistlm}   
\makeatother
\algnewcommand{\parState}[1]{\State%
    \parbox[t]{\dimexpr\linewidth-\algmargin}{\strut #1\strut}}
\usepackage{indentfirst}
\begin{document}
	\title{ Quantum Boosting}

\author{
	Srinivasan Arunachalam\thanks{IBM Research. Srinivasan.Arunachalam@ibm.com} 
	\and
Reevu Maity~\thanks{Clarendon Laboratory, University of Oxford. reevu.maity@physics.ox.ac.uk}
}
\date{}
\maketitle

\begin{abstract}
Suppose we have a weak learning algorithm $\A$ for a Boolean-valued problem: $\A$ produces hypotheses whose bias $\gamma$ is small, only slightly better than random guessing (this could, for instance, be due to implementing $\A$ on a noisy device), can we \emph{boost}  the performance of~$\A$ so that $\A$'s output is correct on $2/3$ of the inputs?
  
    Boosting is a technique that  converts a weak and inaccurate machine learning algorithm into a \emph{strong} accurate learning algorithm. The AdaBoost algorithm by Freund and Schapire (for which they were awarded the G{\"o}del prize in 2003) is one of the widely used boosting algorithms, with many applications in theory and practice. Suppose we have a $\gamma$-weak learner~for a Boolean concept class $\Cc$ that takes time $R(\Cc)$, then the time complexity of AdaBoost scales~as $\VC(\Cc)\cdot  \poly(R(\Cc), 1/\gamma)$, where $\VC(\Cc)$ is the $\VC$-dimension of $\Cc$.
In this paper, we show how \emph{quantum} techniques can improve the time complexity of classical AdaBoost. To this end, suppose we have a $\gamma$-weak \emph{quantum} learner for a Boolean concept class $\Cc$ that takes time $Q(\Cc)$, we introduce a quantum boosting algorithm whose complexity scales as $\sqrt{\VC(\Cc)}\cdot  \poly(Q(\Cc),1/\gamma)$; thereby achieving a quadratic quantum improvement over classical AdaBoost in terms of  $\VC(\Cc)$. 
\end{abstract}


\section{Introduction}
 In the last decade, machine learning (ML) has received tremendous attention due to its success in practice. Given the broad applications of ML, there has been a lot of interest in understanding what are the  learning tasks for which \emph{quantum computers} could provide a speedup. In this direction,  there has been a flurry of quantum algorithms for practically relevant machine learning tasks that theoretically promise either exponential or polynomial quantum speed-ups over classical computers. In the past, theoretical works on quantum machine learning (QML) have focused on developing efficient quantum algorithms with favourable quantum complexities to solve interesting learning problems. More recently, there have been efforts in understanding the interplay between quantum machine learning algorithms and small noisy quantum devices.

The field of QML has given us algorithms for various quantum and classical learning tasks such as (i) quantum improvements to classical algorithms for practically-motivated machine learning tasks such as support vector machines~\cite{quantumsvm}, linear algebra~\cite{prakash:linear}, perceptron learning~\cite{kapoor:perceptronlearning}, kernel-based classifiers~\cite{havlicek:enhancedfeature,li:sublinearalgorithms}, algorithms to compute gradients~\cite{rebentrost:gradient,arunachalam:quantumautoencoders}, clustering~\cite{aimeur:clustering,kerenidis:qmeans}; (ii)  arbitrary quantum states in the PAC setting~\cite{aaronson:qlearnability}, shadow tomography of quantum states~\cite{aaronson:shadow, van2018improvements}, learnability of \emph{quantum objects} such as the class of stabilizer states~\cite{rocchetto:stabilizer}, low-entanglement states~\cite{mithuna:simulata}; (iii) a  quantum framework for learning Boolean-valued concept classes~\cite{bernstein:complexity,bshouty:quantumpac,atici&servedio:qlearning,arunachalam:qexactlearning}; (iv) quantum algorithms for optimization~\cite{harrow:hhl,convex:2,convex:1}; (v) quantum algorithms for machine learning based on generative models~\cite{gao:gan,lloyd2018quGAN}.

While these results seem promising and establish that quantum computers can indeed provide an improvement for interesting machine learning tasks, there are still several practically motivated challenges that remain to be addressed. One important question is whether the assumptions made in some quantum machine learning algorithms are practically feasible? Recently, a couple of works~\cite{chia:dequantize,jethwani:dequantizeSVD} demonstrated that under certain assumptions QML algorithms can be dequantized. In other words, they showed the existence of  \emph{efficient classical} algorithms for machine learning tasks which were previously believed to provide exponential quantum speedups. In this paper we address another important question which is motivated by practical implementation of QML~algorithms:
\begin{quote}
    Suppose $\A$ is a QML algorithm that is theoretically designed to perform \emph{very well}. However, when implemented on a noisy quantum computer, the performance of $\A$ is \emph{weak}, i.e., the output of $\A$ {is} correct on a \emph{slightly} better-than-half fraction of the inputs. Can we \emph{boost} the performance of $\A$ so that $\A$'s output is correct on $2/3$ of the inputs?
\end{quote}

 The classical \emph{Adaptive Boosting} algorithm  (also referred to as AdaBoost) due to Freund and Schapire~\cite{freund:boosting} can be immediately used to convert a weak quantum learning algorithm to a strong algorithm. In this paper, we provide a \emph{quantum boosting} algorithm that quadratically improves upon the classical AdaBoost algorithm. Using our quantum boosting algorithm, not only can we convert a weak and inaccurate  QML algorithm into a strong accurate algorithm, but we can do it in time that is quadratically faster than classical boosting~techniques in terms of some of the parameters of the algorithm.

\subsection{Boosting}
We now briefly describe the Probably Approximately Correct (PAC) model of learning introduced by Valiant~\cite{valiant:paclearning}. For every $n\geq 1$, let $\Cc_n \subseteq \{c:\01^n\rightarrow \{-1,1\} \}$ and $\Cc = \bigcup_{n \geq 1} \Cc_n$ be a concept class. For $\gamma>0$, we say an algorithm~$\A$ \emph{$\gamma$-learns}~$\Cc$ in the PAC model if: for every $n\geq 1$, $c \in \Cc_n$ and  distribution $\mathcal{D}:\01^n\rightarrow [0,1]$, given $n$ and labelled examples $(x,c(x))$ where $x\sim \mathcal{D}$, $\A$ outputs $h:\01^n\rightarrow \{-1,1\}$ such that $\Pr_{x\sim \mathcal{D}}[h(x)=c(x)]\geq 1/2+~\gamma$. In the \emph{quantum PAC} model, we allow a \emph{quantum learner} to possess a quantum computer and \emph{quantum examples} $\sum_x\sqrt{\mathcal{D}_x}\ket{x,c(x)}$.  We call $\gamma$ the bias of an algorithm, i.e.,~$\gamma$ measures the advantage over random guessing. We say $\A$ is a \emph{weak learner} (resp.~\emph{strong learner}) if the bias $\gamma$ scales inverse polynomially with $n$, i.e., $\gamma=1/\poly(n)$ (resp.~$\gamma$ is a universal constant independent of $n$, for simplicity we let $\gamma=1/6$). We formally define these notions in Section~\ref{sec:definitions}.

In the early 1990s, Freund and Schapire~\cite{schapire:boostfirst,freund:boostfirst,freund:boosting} came up with a boosting algorithm called \emph{AdaBoost}  that \emph{efficiently} solves the following problem: suppose we are given a weak learner as a black-box, can we use this black-box to obtain a strong learner? The AdaBoost algorithm by Freund and Schapire was one of the few theoretical boosting algorithms that were simple enough to be extremely useful and successful in practice, with applications ranging from game theory, statistics, optimization, biology, vision and speech recognition~\cite{schapire:foundations}. Given the success of AdaBoost in theory and practice, Freund and Schapire  won the G{\"o}del prize in 2003.

\paragraph{AdaBoost algorithm.}
 We now give a sketch of the classical AdaBoost algorithm and provide more details in Section~\ref{sec:classicalAdaBoost}. Let $\A$ be a weak PAC learner for $\Cc=\cup_{n\geq 1}\Cc_n$ that runs in time $R(\Cc)$ and has bias $\gamma>0$, i.e.,~$\A$ does slightly better than random guessing (think of $\gamma$ as inverse-polynomial in $n$). The goal of boosting is the following: for every $n\geq 1$, \emph{unknown} distribution $\mathcal{D}:
\01^n\rightarrow  [0,1]$ and \emph{unknown} concept $c\in \Cc_n$, construct a hypothesis $H:\01^n\rightarrow \01$ that satisfies 
\begin{align}
\label{eq:stronglearningabstract}
\Pr_{x\sim \mathcal{D}}[H(x)= c(x)]\geq \frac{2}{3},
\end{align}
where $[\cdot]$ is the indicator function which outputs $1$ if $H(x)=c(x)$ and outputs $0$ otherwise.  
AdaBoost algorithm by Freund and Schapire produces such an $H$ by invoking~$\A$ polynomially many times. The algorithm works as follows: it first obtains~$M$ different labelled examples $S=\{(x_i,c(x_i)):i\in [M]\}$ where $x_i\sim \mathcal{D}$ and then AdaBoost is an iterative algorithm that runs for $T$ steps (for some $M,T$ which we specify later). Let $D^1$ be the uniform distribution on $S$. At the $t$th step, AdaBoost defines a distribution $D^t$ depending on $D^{t-1}$ and invokes $\A$ on the  training set $S$ and distribution~$D^t$. Using the output hypothesis $h_t$ of $\A$, AdaBoost computes the \emph{weighted error}
\begin{align}
\label{eq:defnofepst}
\varepsilon_t=\Pr_{x\sim D^t}[h_t(x)\neq c(x)],
\end{align} 
which is the probability of $h_t$  misclassifying a randomly selected training example drawn from the distribution $D^t$. 
The algorithm then uses $\varepsilon_t$ to compute a \emph{weight} $\alpha_t=\frac{1}{2}\ln\Big(\frac{1-\varepsilon_t}{\varepsilon_t}\Big)$ and updates the distribution $D^t$ to $D^{t+1}$ as follows
\begin{align}
\label{eq:distributionupdateabstract_a}
    D^{t+1}_x &= \frac{D^{t}_x}{Z_{t}}
          \times \begin{cases} 
          e^{-\alpha_t} & \text{ if } h_{t}(x) = c(x) \\
          e^{\alpha_t} & \text{ otherwise },
          \end{cases}
\end{align}
where $Z_t=\sum_{x \in S} D^t_x \exp(-c(x)\alpha_t h_t(x))$.\footnote{This distribution update rule is also referred to as the \emph{Multiplicative Weights Update Method} (MMUW). See~\cite[Section~3.6]{arora:multiplicative} on how one can cast AdaBoost into the standard MMUW framework.}  After $T$ iterations, the algorithm outputs the hypothesis $H$
$$
 H(x) = \sign \Big( \sum_{t=1}^T \alpha_t h_t(x) \Big),
$$
where $\alpha_t$ is the weight and $h_t$ is the weak hypothesis computed in the $t$th iteration.\footnote{Note that without loss of generality, we can assume $\sum_t\alpha_t=1$ since renormalizing $\alpha_t$s will not change $H$.} 

It remains to answer three important questions: (1) What is $T$, (2) What is $M$, (3) Why does $H$ satisfy Eq.~\eqref{eq:stronglearningabstract}? The punchline of AdaBoost is the following: by selecting the number of iterations $T=O(\log M)$, the hypothesis $H$ satisfies $H(x)=c(x)$ for every $x\in S$. However, note that this does not imply that $H$ is a strong hypothesis, i.e., it is not clear if $H$ satisfies Eq.~\eqref{eq:stronglearningabstract}. Freund and Schapire showed that, if the number of labelled examples $M$ is at least $O(\VC(\Cc))$ (where $\VC(\Cc)$ is a combinatorial dimension that can be associated with the concept class $\Cc$), then with high probability (where the probability is taken over the randomness of the algorithm and the training set $S$), the final hypothesis $H$ satisfies
$$
\Pr_{x\sim \mathcal{D}} [H(x)= c(x)]\geq 2/3.
$$
In other words, by picking $M$ large enough, not only perfectly classifies every $x\in S$, but is also $2/3$-close to $c$ under $\mathcal{D}$. 
Hence $H$ is a strong hypothesis for the target concept $c$ under the unknown distribution $\mathcal{D}$ which had support on $\01^n$. The overall time complexity of AdaBoost is $\widetilde{O}( n \cdot R(\Cc) \cdot \VC(\Cc))$: the algorithm runs for $T=\log M=\log \VC(\Cc)$ rounds, and in each round  we run a weak learner with time complexity $R(\Cc)$, compute the weighted error $\varepsilon_t$ which takes time $O(M)=O(\VC(\Cc))$ and then update the distributions using arithmetic operations in time $O(n)$.

\subsection{Our results}
The main contribution of this paper is a quantum boosting algorithm that quadratically improves upon the classical algorithm in $\VC(\Cc)$.

\begin{theorem}[Informal]
    For $n\geq 1$ let $\Cc_n \subseteq  \{c:\01^n\rightarrow \{-1,1\} \}$ and  $\Cc = \bigcup_{n \geq 1} \Cc_n$. Let $\A$ be a $\gamma$-weak quantum PAC learner for $\Cc$  that takes time $Q(\Cc)$. Then the quantum time complexity of converting $\A$ to a strong PAC learner is
    $$
    T_Q=\widetilde{O}\Bigg(\sqrt{\VC(\Cc)} \cdot Q(\Cc)^{3/2} \cdot \frac{n^2}{\gamma^{11}}\Bigg).
    $$
\end{theorem}
The classical complexity of AdaBoost scales as $\widetilde{O}(\VC(\Cc) \cdot R(\Cc) \cdot n /\gamma^4)$ where $R(\Cc)$ is the time complexity of a classical PAC learner. Comparing this bound with our main result, we get a quadratic improvement in terms of $\VC(\Cc)$ and also observe that the time complexity of quantum PAC learning $Q(\Cc)$ could be polynomially or even exponentially smaller than classical PAC learning time complexity~$R(\Cc)$.\footnote{In~\cite{arunachalam:optimalpaclearning}, the authors prove that the \emph{sample} complexity of classical and quantum PAC learning is the same up to constant factors, but there exist concept classes demonstrated by~\cite{servediogortler:equiv} for which there could be exponential separations in \emph{time} complexity between quantum and classical learning (under complexity theoretic assumptions).} 

There have been a few prior works~\cite{neven2012qboost,schuld2018quantum,Hsieh:adaboost} which touch upon AdaBoost but none of them rigorously prove that quantum techniques can improve boosting. As far as we are aware, ours is the \emph{first work} that proves quantum algorithms can quadratically improve the complexity of classical~AdaBoost. Given the importance of AdaBoost in classical machine learning, our quadratic quantum improvement could potentially have various applications in~QML. 
 We believe that the $(1/\gamma)$-dependence on our complexity should be  improvable using quantum techniques (and we leave it as an open question).
 Although our complexity is weaker than the classical complexity in terms of $1/\gamma=\poly(n)$, observe that many concept classes have $\VC(\Cc)$ that scales \emph{exponentially} with $n$, in which case our quadratic improvement in terms of $\VC(\Cc)$ ``beats" the ``polynomial loss" (in terms of~$1/\gamma$) in the complexity of  our quantum boosting algorithm.

\paragraph{Applications to NISQ algorithms.}
We now consider the scenario where QML algorithms are implemented on a \emph{noisy-intermediate scale quantum} computer (often referred to as NISQ~\cite{preskill:nisq}): suppose $\A$ is a quantum learner for the following well-known classification problem: given a set~$S$ of training points $x\in \R^d$ labelled as either $0$ or $1$, decide  if the label of an $x\notin S$ is $0$ or $1$. Let us assume $\A$ is \emph{designed} to be a strong learner, i.e., $\A$ outputs a separating hyperplane $H:\R^d\rightarrow \01$ for the classification problem such that, at most a constant $\leq 1/3$-fraction of the points in $\R^d$ are misclassified by $H$. However when implementing $\A$ on a NISQ machine, suppose the errors in the quantum device can only guarantee that $H$ classifies a  $(1/2+\gamma)$-fraction of the points correctly (think of  $\gamma=1/\poly(d)$, i.e., $H$ does barely better than random guessing). Then how do we use the NISQ machine to produce a hyperplane that correctly classifies a 2/3-fraction of the inputs?  Our quantum boosting algorithm can be used here to find a good hyperplane using multiple invocations of the NISQ device. Although our quantum boosting algorithm uses quantum phase estimation as a subroutine, which isn't a NISQ-friendly quantum algorithm, we leave it as an open question if one could use variational techniques as proposed by Peruzzo et al.~\cite{vqeforpe} to replace the quantum phase estimation step.

We now give more details. Let $N$ be a power of $2$, $n=\log_2 N$ and  $\X\subseteq \01^N$. Let $x\in \X$ be an unknown string, which can be thought of as a set of $N$ points labelled $0$ or $1$. Furthermore, let us make the assumption that for a set of $M$ uniformly random  $i_1,\ldots,i_M\in [N]$, our algorithm has knowledge of $S=\{(i_1,x_{i_1}),\ldots,(i_M,x_{i_M})\}$. We believe that this is a realistic assumption, since in most learning algorithms we often have prior knowledge of the  unknown string $x$ at uniformly random coordinates $\{i_1,\ldots,i_M\}$. Let $\A$ be a quantum algorithm that, on input a distribution $ \mathcal{D}:[N]\rightarrow [0,1]$ and $S$, takes time at most~$Q$ to output a string $y\in \01^N$ that satisfies $\Pr_{i\sim \mathcal{D}} [y_i= x_i]\geq \frac{1}{2}+\gamma$. Our quantum boosting algorithm can be used to perform the following: suppose 
$|S|\geq \widetilde{\Omega}\Big(\VC(\X)/\gamma^2\Big),
$ then in time 
$$
\widetilde{O}\Bigg(\sqrt{\VC(\X)}\cdot \frac{Q^{3/2}}{\gamma^{11}} \cdot N^2 \Bigg), 
$$
our quantum algorithm can produce a string $\widetilde{y}$ such that $\Pr_{} [\widetilde{y}_i= x_i]\geq 2/3$ (i.e., the Hamming distance between $\widetilde{y}$ and $x$ is at most $N/3$), where the probability is taken over uniformly random $i\in [N]$.\footnote{Note that we use the uniform distribution here because we started with the assumption that $S=\{(i_1,x_{i_1}),\ldots,(i_M,x_{i_M})\}$ was obtained by uniformly sampling $i_1,\ldots,i_M\in [N]$. Our quantum boosting algorithm works equally well in case obtain indices $i_1,\ldots,i_M$ from an arbitrary (possibly unknown) distribution $D$ instead of the uniform distribution, in which case we obtain a $\widetilde{y}$ such that $\Pr_{i\sim D}[\widetilde{y}_i=x_i]\geq 2/3$.} Classically, one could have used the AdaBoost algorithm to perform the same task, which would have taken time that depends linearly on $\VC(\X)$.

\subsection{Proof sketch} 
We now give a sketch of our quantum boosting algorithm. The quantum algorithm  follows the structure of the classical AdaBoost algorithm. On a very high level, our quantum speedup is obtained by using quantum techniques to estimate the quantity 
$$
\varepsilon_t=\Pr_{x\sim D^t}[h_t(x)\neq c(x)]=\sum_{x\in S} D^t_x \cdot  [h_t(x)\neq c(x)],
$$
quadratically faster than classical methods. In order to do so, use the quantum algorithm for mean estimation, which given a set of numbers $\alpha_1,\ldots,\alpha_M \in [0,1]$, produces an approximation of $\frac{1}{M}\sum_{i\in [M]}\alpha_i$ up to an additive error $\delta$ in time~${\Theta}(\sqrt{M}/\delta)$~\cite{nayakwu:meanestimation,brassard:meanestimation},\footnote{In this section we omit poly-logarithmic factors in the complexity for simplicity.} whereas classical methods take  time~$\Theta(M)$.

\subsubsection{Why does quantum not ``trivially" give a quantum speedup to AdaBoost?}
\label{sec:whynotrivialspeedup}
Although our quantum speedup might seem like an immediate application of quantum mean estimation, using the mean estimation subroutine to improve classical AdaBoost comes with various issues which we highlight  now.
\begin{enumerate}
    \item \textbf{Errors while computing $\varepsilon_t$s}: Quantumly, the mean estimation subroutine \emph{approximates} ${\varepsilon_t}$ up to an additive error $\delta$ in time $O(\sqrt{M}/\delta)$. Suppose we obtain $\varepsilon'_t$ satisfying $|\varepsilon'_t-\varepsilon_t|\leq \delta$. Recall that the distribution update in the $t$th step of AdaBoost is given by
\begin{align}
\label{eq:distributionupdateabstract}
    D^{t+1}_x &= \frac{D^{t}_x}{Z_{t}}
          \times \begin{cases} 
          e^{-\alpha_t} & \text{ if } h_{t}(x) = c(x) \\
          e^{\alpha_t} & \text{ otherwise },
          \end{cases}
\end{align}
    where $Z_t=\sum_{x\in S} D^t_x \exp(-c(x)\alpha_t h_t(x))$ and $\alpha_t=\frac{1}{2}\ln((1-\varepsilon_t)/\varepsilon_t)$. Given an  additive approximation $\varepsilon'_t$ of $\varepsilon_t$, first note that the approximate weights    $\alpha'_t=\frac{1}{2}\ln((1-\varepsilon'_t)/ \varepsilon'_t)$ could be very far from $\alpha_t$. Moreover, it is not clear why $\widetilde{D}^{t+1}$ defined as
    $$
    \widetilde{D}^{t+1}_x=\frac{1}{Z_t}\cdot D^t_x \exp(\alpha_t' c(x)\cdot h_t(x))
    $$
    is  \emph{even close} to a distribution. Another possible way to update our distribution would be
    \begin{align}
        \label{eq:distupdateex}
    \widetilde{D}^{t+1}_x=\frac{1}{Z'_t}\cdot D^t_x \exp(-\alpha_t' c(x)\cdot h_t(x)),
    \end{align}
    where $Z'_t=\sum_{x\in S} D^t_x \exp(-c(x)\alpha'_t h_t(x))$, so by definition $\widetilde{D}^{t+1}$ in Eq.~\eqref{eq:distupdateex} is a distribution. However, in this case note that a quantum learner cannot exactly compute $Z'_t$ in time $O(\sqrt{M})$ but instead can  \emph{approximate} $Z'_t$ and we face the same issue as mentioned above.\footnote{Note that computing $Z'_t$ would take time $O(M)$ classically even though we have knowledge of $\alpha'_t$ since $Z'_t$  involves a summation of $M$ terms.  Hence, we need to approximate $Z'_t$ again using a mean estimation algorithm.}

    \item \textbf{Strong approximation of $\varepsilon_t$:} One possible way to get around this would be to estimate  $\varepsilon_t$ \emph{very well} so that one could potentially show that $\widetilde{D}^{t+1}$ is close to a distribution. However,  it is not too hard to see that if $\widetilde{D}^{t+1}$ should be close to a distribution, then we require a $\delta=1/\sqrt{M}$-approximation of $\varepsilon_t$. Such a strong approximation increases the complexity from $O(\sqrt{M})$ to $O(M)$ which removes the entire quantum speedup.
    
    \item \textbf{Noisy inputs to a quantum learner:} Let us further assume that we could spend time $O(M)$ as mentioned above to estimate $\varepsilon_t$ very well (instead of using classical techniques to compute~$\varepsilon_t$). Suppose we obtain $\widetilde{D}^{t+1}$ which is close to a distribution. Recall that the input  to a quantum learner should be copies of a quantum state $\ket{\psi_t}=\sum_{x\in S}\sqrt{D^t_x}\ket{x,c(x)}$. However, we only have access to a quantum state $\ket{\phi_t}=\sum_{x\in S}\sqrt{\widetilde{D}^t_x}\ket{x,c(x)} +\ket{\chi_t}$, where $\ket{\chi_t}$ is orthogonal to the first part of $\ket{\phi_t}$ (note that  $\ket{\phi_t}$ is no longer a quantum state without the additional quantum register $\ket{\chi_t}$). Now it is unclear what will be the output of a quantum learner on input $\ket{\phi_t}$ instead of $\ket{\psi_t}$.
    
    \item \textbf{Why is the final hypothesis  good:} Assume for now that we are able to show that the learner on input copies of $\ket{\phi_t}$ produces a weak hypothesis $h_t$ for the target concept~$c$. Then, after $T$ steps of the quantum boosting algorithm, the final hypothesis would be $H(x)= \sign\Big(\sum_{t=1}^T \alpha'_t h_t(x)\Big)$. It is not at all  clear why $H$ should satisfy $H(x)=c(x)$ for even a constant fraction of the $x$s in $S$. Observe that the analysis of classical AdaBoost crucially used that ${H}(x)=c(x)$ for almost every $x\in S$ in order to conclude that the generalization error is small, i.e., $\Pr_{x \sim \mathcal{D}} [H(x)\neq c(x)]\leq 1/3$ where $\mathcal{D}$ is an unknown distribution over $\01^n$.
    \end{enumerate}

In this paper, our main contribution is a \emph{quantum boosting algorithm} that overcomes all the issues mentioned above. 

\subsubsection{Quantum boosting algorithm} 
 We now give more details of our quantum boosting algorithm. In order to avoid the issues mentioned in the previous section, our main technical contribution is the following:  we provide a quantum algorithm that modifies the standard distribution update rule of classical AdaBoost in order to take care of the \emph{ approximations} of $\varepsilon_t$s. We show that the output of our modified quantum boosting algorithm  has the same guarantees as classical AdaBoost.
 
Before we elaborate more on the quantum boosting algorithm, we remark that the modified distribution update rule is also applicable in classical AdaBoost. 
Suppose in classical AdaBoost, we obtain the approximations $\varepsilon'_t$s instead of the \emph{exact} weighted errors $\varepsilon_t$s in time $P$. Then our \emph{robust classical AdaBoost} algorithm (i.e., AdaBoost with modified distribution update) can still produce a hypothesis $H$ that has small training error and the complexity of such a robust classical AdaBoost algorithm will be proportional to  $O(P)$. Clearly, it is possible that $P$ could be \emph{much smaller} than~$M$ (which is the time taken by classical AdaBoost to compute $\varepsilon_t$ exactly) in which case the robust classical AdaBoost algorithm is faster than standard classical AdaBoost (in fact we are not aware if the classical AdaBoost or the MMUW algorithm is robust to errors). 

We now discuss the important modification  in our quantum boosting algorithm: the distribution update step. As mentioned before, classically one can compute the quantity $\varepsilon=\Pr_{x\sim D}[h(x)\neq c(x)]$ in time $O(M)$. Quantumly, we describe a subroutine that for a fixed $\delta$, performs the following: outputs `yes' if $\varepsilon\geq \Omega((1-\delta)/(QT^2))$ and `no' otherwise. In the `yes' instance when $\varepsilon$ is large, the algorithm also outputs an approximation $\varepsilon'$ that satisfies $|\varepsilon' - \varepsilon | \leq \delta \varepsilon'$ and in the `no' instance, the algorithm outputs an $\varepsilon'$ that satisfies $|\varepsilon' - \varepsilon | \leq 1/(Q T^2)$. The essential point here is the subroutine takes time $O(\sqrt{M})$.\footnote{We remark that the subroutine outputs `yes' or `no' with high probability (here, for simplicity in exposition, we assume that the subroutine always correctly outputs `yes' or `no').} The subroutine crucially uses the fact that in the `yes' instance,  the complexity of the standard quantum mean estimation algorithm scales as $O(\sqrt{M})$. However, in the `no' instance when $\varepsilon$ is ``small", obtaining a good multiplicative approximation of $\varepsilon'$ using the quantum mean estimation algorithm could potentially take time $O(M)$. In this case, we observe that we \emph{do not} need a good approximation of $\varepsilon$ and instead we set $\varepsilon'=\tau=1/(QT^2)$. We justify this~shortly.

Depending on whether we are in the `yes' instance or `no' instance of the subroutine, we update the distribution differently. In the `yes' instance, we make a distribution update that resembles the standard AdaBoost update using the approximation $\varepsilon'_t$ instead of $\varepsilon_t$. We let $Z_t = 2 \sqrt{\varepsilon'_t (1 - \varepsilon'_t)}$, $\alpha'_t=\frac{1}{2}\ln\Big(\frac{1-\varepsilon'_t}{\varepsilon'_t}\Big)$ and update  $\widetilde{D}^{t}_x$ as follows:
\begin{align}
\label{eq:yesinstanceabstract}
\widetilde{D}^{t+1}_x &= \frac{\widetilde{D}^{t}_x}{(1 + 2 \delta) \hspace{1pt}  Z_{t}}\times
          \begin{cases} 
          e^{-\alpha'_t} & \text{ if } h_{t}(x) = c(x) \\
          e^{\alpha'_t} & \text{ otherwise }.
          \end{cases}
\end{align}
However, in the `no' instance when $\varepsilon_t$ is small, we cannot hope to get a good multiplicative approximation in time $O(\sqrt{M})$. In this case, we crucially observe that $\alpha_t = \frac{1}{2}\ln\Big(\frac{1-\varepsilon_t}{\varepsilon_t}\Big)$ is large, hence with a ``worse approximation" $\varepsilon'_t$ and $\alpha'_t =\frac{1}{2}\ln\Big(\frac{1-\varepsilon'_t}{\varepsilon'_t}\Big)$, we can still show that the hypothesis $H(x)=\sign\Big(\sum_{t=1}^T \alpha'_t h_t(x)\Big)$ has small training error.  As a result, in the `no' instance, we simply let $\varepsilon'_t =\tau$, $Z_t = 2 \sqrt{\tau (1 - \tau )} $ and $ \alpha'_t = \frac{1}{2}\ln\Big(\frac{1-\tau}{\tau}\Big)$ and update $\widetilde{D}^{t}_x$ as follows: 
\begin{align}
\label{eq:noinstanceabstract}
          \widetilde{D}^{t+1}_x &= \frac{\widetilde{D}^{t}_x}{(1 + 2/(QT^2)) Z_{t}}
          \times \begin{cases} 
          (2 - 1/(Q T^2))e^{-\alpha'_t} & \text{ if } h_{t}(x) = c(x) \\
          (1/(QT^2))e^{\alpha'_t} & \text{ otherwise }.
          \end{cases}
\end{align}
Note that the distribution update in Eq.~\eqref{eq:noinstanceabstract} is \emph{not} the standard boosting distribution update and differs from it by  assigning higher weights to the correctly classified training examples and lower weights to the misclassified ones. In both cases of the distribution update  in Eq.~\eqref{eq:yesinstanceabstract},~\eqref{eq:noinstanceabstract}, observe that $\widetilde{D}$ need not be a true distribution. However, we are able to show that $\widetilde{D}$ is very close to a distribution, i.e., we argue that $\sum_{x \in S} \widetilde{D}_x\in [1-30 \delta,1]$. This aspect is very crucial because, in every iteration of the quantum boosting algorithm, we will pass copies of
$$
\ket{\Phi'}=\sum_{x\in S}\sqrt{\widetilde{D}_x}\ket{x,c(x)}+\ket{\chi},\footnote{We need $\ket{\chi}$ because $\widetilde{D}$ is not a distribution, and  $\sum_{x \in S}\sqrt{\widetilde{D}_x}\ket{x,c(x)}$ need not be a valid quantum~state.}
$$
to the quantum learner instead of the ideal quantum state
$$
\ket{\Phi}=\sum_{x\in S}\sqrt{{D}_x}\ket{x,c(x)}.
$$
A priori it is not clear, what will be the output of the weak quantum learner on the input $\ket{\Phi'}$. However, we show that the state $\ket{\Phi'}$ is  close to $\ket{\Phi}$, in particular we show that $|\braket{\Phi'}{\Phi}|\geq 1-\delta$. Suppose a weak quantum learner outputs a weak hypothesis $h$ when given copies of~$\ket{\Phi}$ (with probability at least $1-1/T$). Using the properties of $\widetilde{D}$ and the lower bound on $|\braket{\Phi'}{\Phi}|$, we show that the same quantum learner will output a weak hypothesis $h$ when given copies of the state $\ket{\Phi'}$, with probability at least $1-2/T$. {A union bound over the $T$ iterations of the algorithm shows that with probability at least $2/3$, we obtain a strong hypothesis $H$ after the $T$ iterations}.  Finally, after~$T$ rounds, our quantum boosting algorithm outputs the hypothesis $H(x)=\sign\Big(\sum_{t=1}^T \alpha'_t h_t(x)\Big)$ for all $x\in \01^n$.

It remains to show that the final hypothesis $H$ has small training error. We remark that the calculations to prove this are mathematically technical and are the non-trivial aspects of our quantum~algorithm. Crucially, we use the structure of the modified distribution updates to show that~$H$ has small training error. In order to go from small training error to small generalization error, we use the same ideas as in classical AdaBoost to show that, if the number of classical labelled examples~$M$ is at least $O(\VC(\Cc))$, then $H$ has generalization error at most $1/3$. The overall time complexity of our quantum boosting algorithm is dominated by the subroutine for estimating~$\varepsilon$ in every iteration, which scales as $O(\sqrt{M})$. The remaining part of the quantum boosting algorithm invokes the weak quantum learner which takes time $Q(\Cc)$ and performs arithmetic operations in the distribution update state which takes time~$O(n^2)$. So the overall complexity of our quantum boosting algorithm scales as $\widetilde{O}(n^2 \cdot \sqrt{\VC(\Cc)} \cdot Q(\Cc)^{3/2})$, which is quadratically better than the classical AdaBoost complexity in terms of $\VC(\Cc)$.

\paragraph{Related works.} 
 As we mentioned earlier, there are only a few works that consider boosting in the quantum setting.  Neven et al.~\cite{neven2012qboost} considered a heuristic variant of the AdaBoost algorithm and showed how to implement it using the adiabatic quantum hardware on a D-Wave machine. They give numerical evidence that quantum computers should give a speedup to AdaBoost. Schuld and Petruccione~\cite{schuld2018quantum} consider the problem of \emph{boosting} the performance of quantum classifiers and use AdaBoost as a subroutine. 

Finally, in a recent work Wang et al.~\cite{Hsieh:adaboost} proposed a quantum algorithm to improve the performance of weak learning algorithms.\footnote{In their paper, they consider the setting where the learners are \emph{probabilistic}. We do not discuss that setting here, but our analysis works exactly the same in case the learners are probabilistic.} Their quantum algorithm departs from standard boosting algorithms in several ways since they make various assumptions in their work. With reference to Section~\ref{sec:whynotrivialspeedup}: (1) they assume knowledge of $\varepsilon_t$s \emph{exactly}: note that this is not possible in $o(M)$ time and standard quantum mean estimation takes time in $O(\sqrt{M})$ in order to \emph{approximate} the mean; (2) their quantum algorithm only approximates $\alpha_1,\ldots,\alpha_T$ additively; (3) given such additive approximations, it is not clear why the final output hypothesis  $H$ has small training error. We believe that in order to prove their quantum algorithm outputs a \emph{strong} hypothesis, the time complexity should be $O(M\cdot T^2)$ instead of the claimed $O(\sqrt{M}\cdot T^2)$. However using our techniques, we can improve their complexity to $O(\sqrt{M}T^{5})$. Although this complexity might seem worse than the classical AdaBoost complexity of $O(MT)$, note that we set $T=O(\log M)$ in AdaBoost for the convergence analysis. Hence, the overall quantum complexity is quadratically better than classical AdaBoost in terms of $M$, which is fixed to be $\VC(\Cc)$ in order to have small generalization~error.

\paragraph{Organization.} In Section~\ref{sec:prelim}, we formally define the classical and quantum learning models and state the required claims for our quantum boosting algorithm. In Section~\ref{sec:classicalAdaBoost}, we discuss how classical AdaBoost can improve the performance of weak quantum machine learning algorithms. In Section~\ref{sec:quantumboost}, we finally describe the quantum boosting algorithm which is quadratically faster than the classical AdaBoost algorithm.

\section{Preliminaries}
\label{sec:prelim}

\subsection{Learning definitions}
\label{sec:definitions}

Throughout this paper, we let $[n]=\{1,\ldots,n\}$. Let $c:\01^n\rightarrow \{-1,1\}$. We say $\A$ is given \emph{query} access to $c$ if, $\A$ can \emph{query} $c$, i.e., $\A$ can obtain $c(x)$ for $x$ of its choice. Similar, we say $\A$ has \emph{quantum query} access to $c$, if $\A$ can query $c$ in a superposition, i.e., $\A$ can perform the map
$$
{O}_c:\ket{x,b}\rightarrow \ket{x,c(x)\cdot b}, 
$$
for every $x\in \01^n$ and $b\in \{-1,1\}$. We now introduce our main learning model.

\paragraph{PAC learning.} 	The Probably Approximately Correct (PAC)  model of learning was introduced by   Valiant~\cite{valiant:paclearning} in 1984. A concept class $\Cc$ is a collection of concepts. Often, $\Cc$ is composed of a subclass of functions $\{\Cc_n\}_{n\geq 1}$, i.e., $\Cc=\cup_{n\geq 1}\Cc_n$, where $\Cc_n$ is a
  collection of Boolean functions $c:\01^n\rightarrow \{-1,1\}$, which are often referred to as \emph{concepts}. In the PAC learning model, a learner~$\A$ is given $n\geq 1$ and access to \emph{labelled examples} $(x,c(x))$ where $(x,c(x))$ is drawn according to the unknown distribution $\mathcal{D}:\01^n\rightarrow [0,1]$ and $c\in \Cc_n$ is the \emph{unknown} target concept (which the learner is trying to learn). The goal of $\A$ is to output a hypothesis $h:\01^n\rightarrow \{-1,1\}$ that is $\eta$-close to $c$ under $\mathcal{D}$. We say that $\A$ is an \emph{$(\eta,\delta)$-PAC learner} for a concept class~$\Cc$ if it satisfies: 

	\begin{quote}
     for every $n\geq 1$, $c\in \Cc_n$ and distributions $\mathcal{D}$, $\A$ takes as input $n,\delta,\eta$ and  labelled examples $(x,c(x))$ and with probability $\geq 1-\delta$, $\A$ outputs a hypothesis $h$
    such~that 
    $$\Pr_{x\sim \mathcal{D}}[h(x)\neq c(x)]\leq \eta.$$
	\end{quote}
 The sample complexity and time complexity of a learner is the number of labelled examples and number of bit-wise operations (i.e., time taken)  that suffices to learn $\Cc$ (under the hardest concept $c\in \Cc$ and distribution $\mathcal{D}$).  
 
 Throughout this paper, we assume that all concept classes $\Cc$ are defined as $\Cc=\cup_{n\geq 1}\Cc_n$ where $\Cc_n\subseteq \{c:\01^n\rightarrow \01\}$ (in fact from here onwards, we will assume that $\Cc_n$ is always a subset of $\{c:\01^n\rightarrow \01\}$ and do not explicitly mention it).
 
 In the quantum PAC model, a learner is a \emph{quantum algorithm} given access to the quantum examples $\sum_x\sqrt{\mathcal{D}_x}\ket{x,c(x)}$. The quantum sample complexity is the number of quantum examples used by the quantum learner to learn $\Cc$ (on the hardest $c\in \Cc$ and distribution $\mathcal{D}$) and the \emph{time complexity} of a quantum algorithm is the total number of gates involved (i.e., the number of gates it takes to implement various unitaries during the quantum algorithm) as well as the number of gates it takes to prepare quantum states.  The remaining aspects of the quantum PAC learner is defined analogous to the classical PAC model. For more on this subject, the interested reader is referred to~\cite{arunachalam:quantumsurveylearning}. We now define what it means for an algorithm $\A$ to be a strong and weak learner for a concept class $\Cc$.

	\begin{definition} [Weak  learner]
	\label{def:QAB_weak_learner}
	Let $\Cc=\cup_{n\geq 1}\Cc_n$ be a concept class. We say~$\A$ is a \emph{weak (quantum) learner} for $\Cc$ if it satisfies the following: there exists a polynomial $p$ such that for  all $n\geq 1$, for all $c\in \Cc_n$ and distributions $\mathcal{D}:\01^n\rightarrow [0,1]$, algorithm $\A$, given $n$ and (quantum) query access to $c$, with probability $\geq 2/3$, outputs a hypothesis $h:\01^n\rightarrow \01$ satisfying
	\begin{align}
	    \label{eq:weakhypothesis}
	  \Pr_{x\sim \mathcal{D}} [h(x)=c(x)] \geq \frac{1}{2}+\frac{1}{p(n)}.
	\end{align}
	\end{definition}

    	\begin{definition} [Strong learner]
	\label{def:QAB_strong_learner}
	Let $\Cc=\cup_{n\geq 1}\Cc_n$ be a concept class. We say~$\A$ is a \emph{strong (quantum) learner} for $\Cc$ if it satisfies the following:  for all $n\geq 1$, $c\in \Cc_n$ and distributions $\mathcal{D}:\01^n\rightarrow [0,1]$, algorithm $\A$, given $n$ and (quantum) query access to $c$, with probability $\geq 2/3$, outputs a hypothesis $h:\01^n\rightarrow \01$ satisfying
		\begin{align}
	    \label{eq:stronghypothesis}
	\Pr_{x\sim \mathcal{D}} [h(x)=c(x)] \geq \frac{2}{3}.
	\end{align}

	\end{definition}

Throughout this paper, we will assume that we have classical or quantum query access to the output  hypothesis $h$. Similarly, we say $h$ is a \emph{weak hypothesis} (resp.~\emph{strong hypothesis}) under $\mathcal{D}$ if $h$ satisfies Eq.~\eqref{eq:weakhypothesis} (resp.~Eq.~\eqref{eq:stronghypothesis}). We now define the Vapnik-Chervonenkis dimension (also referred to as VC dimension)~\cite{vapnik:vcdimension}.
	
\begin{definition} ($\VC$ dimension~\cite{vapnik:vcdimension})
 Fix a concept class $\Cc$ over $\01^n$. A set $\Se=\{s_1,\ldots,s_t\}\subseteq \01^n$ is said to be \emph{shattered} by a concept class $\Cc$ if  $\{(c({s_1}) \cdots c({s_t})) : c\in \Cc\} =\{-1,1\}^{t}$. In other words, for every labeling $\ell\in \{-1,1\}^{t}$, there exists a $c\in \Cc$ such that $(c({s_1}) \cdots c({s_t}))=\ell$. The $\VC$ dimension of $\Cc$ (denoted by $\VC(\Cc)$) is the size of the largest $\Se\subseteq \01^n$ that is shattered by~$\Cc$.  
\end{definition}

We now define  two important misclassification errors which we will encounter often. Suppose an algorithm $\A$ is given a set of labelled examples $S=\{(x_1,y_1),\ldots,(x_M,y_M)\}$ where $(x_i,y_i)\in \01^n\times \{-1,1\} $ is drawn from a joint distribution $\mathcal{D}:\01^n\times \{-1,1\} \rightarrow [0,1]$ and suppose $\A$ outputs a hypothesis $h:\01^n\rightarrow \{-1,1\} $.
  The  \emph{training error} of $h$ is defined as the error of $h$ on the \emph{training set} S, i.e., 
$$
\text{training error of } h= \frac{1}{M} \sum_{i=1}^M [h(x_i)\neq y_i]. 
$$
Ultimately, the goal of $\A$ should be to do well on labelled examples $(x,y) \notin S$ where $(x,y)$ is sampled from the same distribution $\mathcal{D}:\01^n\times \{-1,1\} \rightarrow [0,1]$ that generated the training set $S$. In order to quantify the \emph{goodness} of the hypothesis $h$, the \emph{true} error or the \emph{generalization error} is defined as
$$
\text{generalization error of } h= \Pr_{(x,y)\sim \mathcal{D}} [h(x)\neq y]. 
$$

\subsection{Required claims}
In order to prove our main results, we will use the following well-known results. 
  
  \begin{theorem}[Amplitude Amplification~\cite{Brassard:AmpEst}]
   \label{thm:amp_amplification}
   Let $p,a,a'>0$. There is a quantum algorithm $\A$ that satisfies the following: given access to a unitary $U$ such that $U\ket{0}=\ket{\psi}$ where $\ket{\psi}=\sqrt{p}\ket{\psi_0}+\sqrt{1-p}\ket{\psi_1}$ for an \emph{unknown} $p>a$ and  $\ket{\psi_0},\ket{\psi_1}$ are orthogonal quantum states, $\A$ makes an expected number of $\Theta(\sqrt{a'/a})$ queries to $U,U^{-1}$  and outputs~$\ket{\psi_0}$ with probability $a'>0$.
   \end{theorem}

    \begin{theorem}[Amplitude Estimation \cite{Brassard:AmpEst}]
    \label{thm:amplitude_estimation}
   There is a quantum algorithm $\A$ that satisfies the~following: given access to a unitary $U$ such that $U \ket{0}= \ket{\psi}$ where
        $\ket{\psi} = \sqrt{a} \ket{\psi_0} + \sqrt{1-a} \ket{\psi_1}$
    and $\ket{\psi_0}$, $\ket{\psi_1}$ are~orthogonal quantum states, $\A$ makes $M$ queries to $U$ and $U^{-1}$  and with probability $\geq 2/3$, outputs~$\widetilde{a}$ such that
    \begin{align}
    \label{eq:approxofatilde}
        \vert \widetilde{a} - a \vert \leq 2 \pi \frac{\sqrt{a(1-a)}}{M} + \frac{\pi^2}{M^2}.
    \end{align}
    \end{theorem}

\begin{theorem}[Multiplicative amplitude estimation \cite{ambainis2010multiplicativeAE}] 
    \label{thm:multiplicativeamplitude_estimation}
   Let $c\in (0,1]$. There is a quantum algorithm $\A$ that satisfies the following: given access to a unitary  $U$ such that $U \ket{0}= \ket{\psi}$ where
        $\ket{\psi} = \sqrt{a} \ket{\psi_0} + \sqrt{1-a} \ket{\psi_1}$
    and $\ket{\psi_0}$, $\ket{\psi_1}$ are orthogonal quantum states and promised that either $a=0$ or $a\geq p$, with probability $\geq 1-\delta$, $\A$  outputs an estimate $\widetilde{a}$ satisfying
         $|a - \widetilde{a} | \leq c\cdot  \widetilde{a}$ if $a \geq p$, and
        $\widetilde{a} = 0$ if $ a = 0$.
    %
    The total number of queries made by $\A$ to $U$, $U^{-1}$ is 
    \begin{align}
        O\Bigg( \frac{\log (1/\delta)}{c} \Bigg(1 + \log \log \frac{1}{p} \Bigg) \sqrt{\frac{1}{\max \{a,p \}}} \Bigg).
 \end{align} 
    \end{theorem}

    \section{Classically boosting quantum machine learning algorithms}
    \label{sec:classicalAdaBoost}
    
    In this section, we describe the classical AdaBoost algorithm and explain how one can use AdaBoost to improve a weak quantum learner to a strong quantum learner. 
    \subsection{Classical AdaBoost}
    We begin with presenting the classical AdaBoost algorithm. AdaBoost achieves the following goal: suppose $\A$ is a $\gamma$-weak PAC learner for a concept class $\Cc$ (think of $\gamma = 1/ \poly(n)$).
    Then, for a fixed unknown distribution~$\mathcal{D}$, can we use $\A$ multiple times to output a hypothesis $H$ such that $\Pr_{x\sim \mathcal{D}}[H(x)=c(x)]\geq 2/3$?  Freund and Schapire~\cite{freund:boosting} gave the following simple algorithm that outputs such a strong hypothesis.
    \begin{algorithm}[H]
   		\caption{Classical AdaBoost}
    \label{alg:classicalAdaBoost}
		\textbf{Input:} Classical weak learner $\mathcal{A}$, query access to training samples $S =\{(x_1,c(x_1)), \ldots, (x_M,c(x_M))\}$, where $x_i\sim \mathcal{D}$ and $\mathcal{D}:\01^n\rightarrow [0,1]$ is an \emph{unknown} distribution.
		\\[1mm]
\textbf{Initialize:} Let $D^1$ be the uniform distribution over $S$.

		\begin{algorithmic}[1]
		 \For{t = 1 to $T$  }
		 \parState{%
		 Train a weak learner $\mathcal{A}$ on the distribution $D^t$ using the labelled examples $S$. Suppose we obtain a hypothesis $h_t$.
		 }
		 
		 \State Compute the weighted error $\varepsilon_t = \sum_{x\in S} D^t_x [h_t(x)\neq c(x)]$, and let $\alpha_t = \frac{1}{2} \ln \Big( \frac{1 - \varepsilon_t }{\varepsilon_t} \Big)$. 
		
		 \State Update the distribution $D^{t}_x$  as follows:
		  \begin{align*}
          D^{t+1}_x &= \frac{D^{t}_x}{Z_{t}}
          \times \begin{cases} 
          e^{-\alpha_t} & \text{ if } h_{t}(x) = c(x) \\
          e^{\alpha_t} & \text{ otherwise }
          \end{cases}
           \\
           &= \frac{D^t_x \exp{\big(-c(x) \alpha_t h_t(x) \big)} }{Z_t},
           \end{align*}
\qquad  where $Z_t=\sum_{x \in S} D^t_x \exp{\big(-c(x) \alpha_t h_t(x))} $.
		 \EndFor	
		\end{algorithmic}
				\textbf{Output:} Hypothesis $H$ defined as $H(x) = \sgn \Big( \sum_{t=1}^T \alpha_t h_t(x) \Big)$ for all $x\in \01^n$.
			\end{algorithm}
	We do not prove why the output hypothesis~$H$ is a strong hypothesis  and simply state the main result of Freund and Schapire~\cite{freund:boosting,schapire:foundations}. Suppose $T\geq (\log M)/\gamma^2$, then the output~$H$ of Algorithm~\ref{alg:classicalAdaBoost} with high probability (over the randomness of the algorithm and the training set~$S$) has \emph{zero training error} (i.e., $H(x_i)=c(x_i)$ for every $i\in [M]$). A priori, it might seem that this task is easy to achieve, since we could simply construct a function $g$ that satisfies $g(x_i)=c(x_i)$ for every $i\in [M]$ given explicit access to $S=\{(x_i,c(x_i))\}_{i \in [M]}$. However, recall that the goal of AdaBoost is to output a strong hypothesis $H$ that satisfies $\Pr_{x\sim \mathcal{D}}[H(x)=c(x)]\geq 2/3$ (where $\mathcal{D}$ is the unknown distribution according to which $S$ is generated in Algorithm~\ref{alg:classicalAdaBoost}) and it is not clear whether $g$ satisfies this condition. Freund and Schapire~\cite{freund:boosting} showed the surprising property that the output of classical AdaBoost $H$ (i.e., a weighted combination of the hypotheses generated in each iteration) is in fact a strong hypothesis, provided $M$ is \emph{sufficiently large}. In particular, Freund and Schapire~\cite{freund:boosting} proved the following theorem.
    
    \begin{theorem} 
    [{\cite[Theorems~4.3 and 4.6]{schapire:foundations}} ]
    \label{thm:goingfromtrainingtogeneralization}
    Fix $\eta,\gamma>0$. Let $\Cc=\cup_{n\geq 1}\Cc_n$ be a concept class and~$\A$ be a~$\gamma$-weak PAC learner for $\Cc$ that takes time $R(\Cc)$. Let $n\geq 1$, $\mathcal{D}:\01^n\rightarrow [0,1]$ be an unknown distribution, $c\in \Cc_n$ be the unknown target concept and 
   $$
   M = \Bigg\lceil\frac{\VC(\Cc)}{\gamma^2}\cdot \frac{\log(\VC(\Cc)/\gamma^2)}{\eta^2 }\Bigg\rceil.
   $$
   Suppose we run Algorithm~\ref{alg:classicalAdaBoost} for $T\geq ((\log M)\cdot \log(1/\delta))/(2\gamma^2)$ rounds, then with probability $\geq 1-\delta$ (over the randomness of the algorithm and the random examples $\{(x_i,c(x_i))\}_{i \in [M]}$ where $(x_i,c(x_i))\sim \mathcal{D}$), we obtain a hypothesis $H$ that has zero training error\footnote{Suppose $H$ has training error $\delta$, then the generalization error  becomes $   \Pr_{x\sim \mathcal{D}} [H(x)=c(x)]\geq 1-\delta-\eta$.} and small generalization error
   $$
   \Pr_{x\sim \mathcal{D}} [H(x)=c(x)]\geq 1-\eta.
   $$
   Moreover the time complexity of the classical AdaBoost algorithm is 
   $$
   \widetilde{O}(R(\Cc)\cdot TM n)=\widetilde{O}\Bigg(\frac{\VC(\Cc)}{\eta^2} \cdot R(\Cc) \cdot \frac{n}{\gamma^4}\cdot\log(1/\delta)\Bigg).
   $$
\end{theorem}

\subsection{Boosting quantum machine learners}
We now describe how classical AdaBoost can improve the performance of \emph{quantum} machine learning algorithms. Suppose $\A$ is a quantum PAC learner that learns a concept class $\Cc$ in time $Q(\Cc)$. Can we use $\A$ multiple times to construct a strong quantum learner? Indeed, this is possible using classical AdaBoost which we describe below.

\begin{algorithm}[H]
   		\caption{Classical boosting of weak quantum learners}
    \label{alg:classical-quantumAdaBoost}
		\textbf{Input:} \emph{Quantum} weak learner $\mathcal{A}$,  \emph{quantum} query access to training samples $S =\{(x_1,c(x_1)), \ldots, (x_M,c(x_M))\}$ where $x_i\sim \mathcal{D}$ and $\mathcal{D}:\01^n\rightarrow [0,1]$ is an unknown distribution.
		\\[1mm]
\textbf{Initialize:} Let $D^1$ be the uniform distribution over $S$.

		\begin{algorithmic}[1]
		 \For{t = 1 to $T$  (assume classical query access to $h_1,\ldots,h_{t-1}$ and knowledge of $D^1,\ldots,D^{t-1}$).}
		 \vspace{0.3mm}
		 
		 \State Prepare  $Q(\Cc)$ copies of $\ket{\psi_t   }=\sum_{x\in S} \sqrt{D^t_x}\ket{x}$. 
		 \vspace{0.3mm}
		
		 \State For every $\ket{\psi_t}$, make a quantum query to $S$ to produce $\ket{\psi'_t}=\sum_{x\in S} \sqrt{D^t_x}\ket{x,c(x)}$. 
		 \vspace{0.3mm}
		
		 \State Pass $\ket{\psi'_t}^{\otimes Q(\Cc)}$ to $\A$ and suppose $\A$ produces a hypothesis $h_t$.
		 
	 \vspace{0.3mm}
		
		 
		 \State Compute the weighted error $\varepsilon_t = \sum_{x\in S} D^t_x \cdot [h_t(x)\neq c(x)]$, and let $\alpha_t = \frac{1}{2} \ln \Big( \frac{1 - \varepsilon_t }{\varepsilon_t} \Big)$. 
		 %
		 \vspace{0.3mm}
		
		 \State Update the distribution $D^{t}_x$  as follows:
		  \begin{align*}
          D^{t+1}_x &= \frac{D^{t}_x}{Z_{t}}
          \times \begin{cases} 
          e^{-\alpha_t} & \text{ if } h_{t}(x) = c(x) \\
          e^{\alpha_t} & \text{ otherwise }
          \end{cases}
           \\
           &= \frac{D^t_x \exp{\big(-c(x) \alpha_t h_t(x) \big)} }{Z_t},
           \end{align*}
	         where $Z_t=\sum_{x \in S} D^t_x \exp{\big(-c(x) \alpha_t h_t(x))} $.
		 \EndFor	
		\end{algorithmic}
				\textbf{Output:} Hypothesis $H$ defined as $H(x) = \sgn \Big( \sum_{t=1}^T \alpha_t h_t(x) \Big)$ for all $x\in \01^n$.
	\end{algorithm}
The correctness of this algorithm  follows from the classical AdaBoost algorithm, since at each iteration, the output $h_t$ of the quantum learner is the same as classical AdaBoost. We simply analyze the time complexity of the algorithm here: in the worst-case, step~$2$ takes time $O(n^2MQ(\Cc))$ (note that one could potentially use tricks by Grover-Rudolph~\cite{groverrudolph:state}, Kaye-Mosca~\cite{kayemosca:state} to prepare $\ket{\psi_t}$ more efficiently), step~$3$ takes one quantum query and time $O(n \log M)$ (assuming we have an efficiently implementable quantum random-access-memory (QRAM), we discuss this further in Section~\ref{sec:describeboosting} after Eq.~\eqref{eq:QAB_distribution_quantumupdateC}), step~$4$ takes time~$Q(\Cc)$ (by assumption of $\A$), step $5$ involves making $M$ queries to $h_t$ and $O(n)$ operations to compute $\varepsilon_t,\alpha_t$ and step~$6$ takes time $O(n M)$ in order to update the distribution. Overall Algorithm~\ref{alg:classical-quantumAdaBoost} takes time $O(n^2 \cdot MQ(\Cc) \cdot T )$ time. Crucially, in Algorithm~\ref{alg:classical-quantumAdaBoost} a quantum computer is required  only to prepare the state $\ket{\psi_t}$ in every step and run the weak quantum learner. The remaining computation can be done on a classical device.
Putting everything together, we obtain the following theorem whose proof follows from Theorem~\ref{thm:goingfromtrainingtogeneralization}.

\begin{theorem}
    Fix $\eta,\gamma>0$. Let $\Cc=\cup_{n\geq 1}\Cc_n$ be a concept class and $\A$ be a $\gamma$-weak quantum PAC learner for $\Cc$ that takes time $Q(\Cc)$. Let $n\geq 1$ and  $c\in \Cc_n$ be the unknown target concept. Then the time complexity of  Algorithm~\ref{alg:classical-quantumAdaBoost} to produce a strong hypothesis~is 
    $$
    \widetilde{O}\Bigg(\frac{\VC(\Cc)}{\eta^2}\cdot Q(\Cc) \cdot \frac{n^2}{\gamma^4}\Bigg),
    $$
    Moreover, Algorithm~\ref{alg:classical-quantumAdaBoost} uses the quantum computer only for state preparation (in Step $2$) and to run the weak-quantum PAC learning  algorithm $\A$ (in Step $4$) and the remaining operations can be performed on a classical device.
\end{theorem}

\section{Quantum Boosting}
\label{sec:quantumboost}

In the previous section, we used \emph{classical} techniques to boost a weak quantum learner to a strong quantum learner. In this section, we use \emph{quantum} techniques in order to improve the time complexity of AdaBoost. Like in classical AdaBoost, we break the analysis into two stages. Stage (1) is a quantum algorithm that reduces training error: produces a hypothesis that does well on the training set and Stage (2) reduces generalization error: we show that for a sufficiently large training set, not only does the hypothesis output in Stage (1) has a small training error, but also has a small generalization error. In Section~\ref{sec:suboptimalsection}, we first present a quantum algorithm for boosting and in Section~\ref{sec:generalizationofquantum}, we show why the hypothesis generated in Section~\ref{sec:suboptimalsection} is a strong~hypothesis.

\subsection{Reducing the training error}
\label{sec:suboptimalsection}
The bulk of the technical work in our quantum boosting algorithm lies in reducing the training error and we devote this entire section to proving it. We now state the main theorem for Stage~(1) of our quantum boosting algorithm. 

\begin{theorem} 
    \label{thm:quantum_AdaBoost}
    Let $\gamma>0$. Let $\Cc=\cup_{n\geq 1}\Cc_n$ be a concept class and $\A$ be a $\gamma$-weak quantum PAC learner for~$\Cc$ that takes time~$Q(\Cc)$. Let $n\geq 1$,   $\mathcal{D}:\01^n\rightarrow [0,1]$ be an unknown distribution, $c\in \Cc_n$ be the unknown target concept and $M$ be sufficiently large.\footnote{We quantify what we mean by \emph{sufficiently large} in the next section, in particular in Theorem~\ref{finalcomplexityofquantumAdaBoost}.} Given a training set $S=\{(x_i,c(x_i))\}_{i \in [M]}$ where $x_i \sim \mathcal{D}$ and $c\in \Cc_n$, our quantum boosting algorithm takes time $\widetilde{O}(n^2 \sqrt{M}\cdot Q(\Cc)^{3/2})$,\footnote{Here $\widetilde{O}(\cdot)$ hides poly-logarithmic factors in $M$.} and with probability $\geq 2/3$, outputs a hypothesis~$H$ that has training error at most $1/10$.
	\end{theorem}
Since our quantum algorithm is fairly involved to describe and analyze, we break down this section into four subsections, in order to separate the technicalities from our quantum algorithm. In Section~\ref{sec:describeboosting}, we describe our quantum boosting algorithm. In Section~\ref{sec:claimsboosting}, we prove a few claims which are unproven in Section~\ref{sec:describeboosting}. In Section~\ref{sec:proofofboosting}, we analyze the correctness of the algorithm and finally in Section~\ref{sec:compofboosting}, we analyze the time complexity of our quantum boosting algorithm. Putting together these four subsections gives a proof of Theorem~\ref{thm:quantum_AdaBoost}.

\subsubsection{Quantum boosting algorithm for reducing training error}
\label{sec:describeboosting}
    We begin by presenting our quantum algorithm. For simplicity in notation, we first  denote~$Q(\Cc)$ as $Q$ and secondly we assume that $Q$ is the time it takes for $\A$ to output a weak hypothesis with probability at least $1-O(1/T)$ (note that this  only increases the complexity by a multiplicative $O(\log T)$-factor). Since the sample complexity of $\A$ is at most the time complexity, we will assume that it suffices to provide $\A$ with $Q$ quantum examples. 
    Our quantum algorithm is a $T$-round iterative algorithm similar to classical AdaBoost and in each round, our quantum algorithm produces a distribution $\widetilde{D}$. In the $t$th round, our quantum algorithm follows a three-step~process: 
   
     \begin{enumerate}
        \item Invoke the weak quantum learner $\A$ to produce a weak hypothesis  $h_t$ under an \emph{approximate} distribution $\widetilde{D}^t$ over the training set $S$. 
        \item By making quantum queries to $h_t$, our algorithm computes $\varepsilon'_t$, an approximation to  $\widetilde{\varepsilon_t}=\Pr_{x\sim \widetilde{D}^t} [h_t(x)\neq c(x)]$. We then use $\varepsilon'_t$ to update the distribution $\widetilde{D}^t$ to~$\widetilde{D}^{t+1}$. In this step, we depart from standard AdaBoost.
        \item Using $\varepsilon'_t$ compute a \emph{weight} $\alpha'_t$. After $T$ steps, output a hypothesis $H(x)=\sign\Big(\sum_{t=1}^T \alpha'_t h_t(x) \Big)$.
    \end{enumerate}
   Before we describe our quantum algorithm, we first describe a subroutine which will be useful in performing step (2) in the $3$-step procedure above. This subroutine uses ideas developed by Ambainis~\cite{ambainis2010multiplicativeAE} and the proof uses ideas required to prove Theorem~\ref{thm:multiplicativeamplitude_estimation}.  Let $\widetilde{\varepsilon}, M>0$. 
    
    \begin{algorithm}
    \caption{Modified Amplitude Estimation}
    \label{alg:multiplicativeamplitudeestimation}
     \textbf{Input:} The state $\ket{\psi} =  \sqrt{\widetilde{\varepsilon}/M} \ket{\phi_1} \ket{1}+ \sqrt{1 - \widetilde{\varepsilon}/M} \ket{\phi_0} \ket{0} $ and the unitary $U$ such that $U \ket{0}=\ket{\psi}$.
    \vspace{1mm}

	 \begin{algorithmic}[1]
	  \For{$J$ = $\frac{2 \pi \sqrt{M}}{\delta}$ to $\frac{16 \sqrt{2} \pi \sqrt{M}} {\delta }\cdot \sqrt{QT^2} \log (MT/\delta)$ }
	  \vspace{2mm}
	  
	  \parState{%
	  Let $\varepsilon'/M$ be the output after performing amplitude estimation (in Theorem~\ref{thm:amplitude_estimation}) to estimate  $\widetilde{\varepsilon}/M$ using $J$ queries to $U$ and $U^{-1}$.
	  }
	  \vspace{1mm}
	  \parState{%
	  Check if 
	  $
	  \frac{2\sqrt{2} \pi \sqrt{(1 - \delta) \varepsilon'}}{J\sqrt{M}} + \frac{\pi^2}{J^2} \leq  \frac{ \delta  \varepsilon'}{M}.
	  $
	  If yes, then output $\varepsilon'$ and quit the loop. Else, let $J = 2 \cdot J$.
	  } 
	  \EndFor
	 \end{algorithmic}
	 
	 \textbf{Output:} 
	 $\{\varepsilon', \text{ yes}\}$ if there exists $\varepsilon'$ in step $(3)$, else output 
	$\{\varepsilon' = 1/(QT^2), \text{ no}\}$. 
    \end{algorithm}

    \begin{lemma}
    \label{lem:updatingapproxdistribution}
    Let $\delta = 1/(10 Q T^2)$. Algorithm~\ref{alg:multiplicativeamplitudeestimation} satisfies the following: with probability $\geq 1-10\delta/T$, if the output is $\{\varepsilon',\text{ yes}\}$,
    then $| \widetilde{\varepsilon} - \varepsilon' | \leq \delta \varepsilon'$; and if the output is $\{\varepsilon' = 1/(QT^2), \text{ no}\}$, then $| \widetilde{\varepsilon} - \varepsilon' | \leq 1/(QT^2)$. The total number queries to $U $ and $U^{-1}$ used by Algorithm~\ref{alg:multiplicativeamplitudeestimation} is $O(\sqrt{M}Q^{3/2}T^3)$.
    \end{lemma}

    \begin{proof}
     We first consider the case when Algorithm~\ref{alg:multiplicativeamplitudeestimation} outputs $\{\varepsilon',\text{ yes}\}$. In this case, there exists a~$J$ and $\varepsilon'$ which satisfies the relation in step (3) of the algorithm. First observe that, since $\varepsilon'/M$ was obtained by amplitude amplification in step $(2)$, we have 
    \begin{align}
    \label{eq:lemmbound}
    \Bigg| \frac{\varepsilon'}{M} - \frac{\widetilde{\varepsilon}}{M} \Bigg| \leq \Bigg| \frac{\varepsilon'}{M} - \frac{(1 - \delta) \widetilde{\varepsilon }}{M}  \Bigg|  \leq \frac{2 \pi \sqrt{(1 - \delta) \widetilde{\varepsilon}}}{J\sqrt{M}} + \frac{\pi^2}{J^2}  \leq \frac{2\sqrt{2} \pi \sqrt{(1 - \delta) \varepsilon'}}{J\sqrt{M}} + \frac{\pi^2}{J^2},
    \end{align}
    where the second inequality used Eq.~\eqref{eq:approxofatilde} in Theorem~\ref{thm:amplitude_estimation} and the third inequality used the first and second inequalities to conclude
    $ | \widetilde{\varepsilon} - \varepsilon' |  \leq \frac{2 \pi \sqrt{(1 - \delta) M \widetilde{ \varepsilon}}}{J} + \frac{\pi^2 M}{J^2}$. This implies
    $$\widetilde{\varepsilon} \leq \varepsilon' + \frac{2 \pi \sqrt{(1 - \delta) M \widetilde{ \varepsilon}}}{J} + \frac{\pi^2 M}{J^2} \leq 2 \varepsilon',
    $$
    where we used $J \geq (2 \pi \sqrt{M})/ \delta$. Putting together the upper bound in Eq.~\eqref{eq:lemmbound} along with the upper bound in Step (3) of the algorithm, we get $| \widetilde{\varepsilon} - \varepsilon' | \leq \delta \varepsilon'$.
    Furthermore, we also show that $\widetilde{\varepsilon}\geq (1-2\delta)/(64 QT^2)$. Recall that $J,\varepsilon'$ satisfies step $(3)$ of the algorithm. In particular, this implies that
    $$
    \frac{2\sqrt{2} \pi \sqrt{(1 - \delta) \varepsilon'}}{J_{max}\sqrt{M}} + \frac{\pi^2}{J_{max}^2} \leq   \frac{ \delta  \varepsilon'}{M},
    $$
    since $J\leq J_{max}$. Substituting the value of $J_{max}$ in the inequality above gives
    $ \frac{ \sqrt{(1 - \delta)\varepsilon'}}{8 \sqrt{Q T^2}} + \frac{\delta }{512 Q T^2} \leq    \varepsilon'. $
    Solving for the above equation, we obtain $\varepsilon' \geq 1/(64 Q T^2) \cdot (1 - \delta)$ (we ignore the other solution for~$\varepsilon'$ since $\varepsilon' \geq 0$). Using $| \widetilde{\varepsilon} - \varepsilon' | \leq \delta \varepsilon'$, we get $\widetilde{\varepsilon} \geq (1-\delta) \varepsilon'  \geq (1 - 2 \delta )/(64 Q T^2) $.

    Now we consider the case when Algorithm~\ref{alg:multiplicativeamplitudeestimation} outputs $\{\varepsilon' = 1/(QT^2),\text{ no}\}$, and we argue that $| \widetilde{\varepsilon} - \varepsilon' | < 1/(Q T^2)$. In order to see this, first observe that
    \begin{align}
    \label{eq:bound_noinstance}
    \Bigg| \frac{\varepsilon'}{M} - \frac{\widetilde{\varepsilon}}{M} \Bigg| \leq \frac{2\sqrt{2} \pi \sqrt{(1 - \delta) \varepsilon'}}{J\sqrt{M}} + \frac{\pi^2}{J^2} \leq \frac{\delta \sqrt{2 \varepsilon' } }{M} + \frac{\delta^2}{4 M } \leq \frac{10 \delta}{M},
    \end{align}
    where the first inequality used Eq.~\eqref{eq:lemmbound}, the second inequality used $J  \geq (2 \pi \sqrt{M})/ \delta$ and the third inequality used $\varepsilon' < 1$. Using $\delta = 1/(10 Q T^2)$, we obtain $| \widetilde{\varepsilon} - \varepsilon' | \leq 1/(Q T^2)$. Furthermore, we show that in the `no' instance, we have  $\widetilde{\varepsilon} < 1/(Q T^2)$. We prove this by a contrapositive argument: suppose  $\widetilde{\varepsilon} \geq 1/(Q T^2)$, there exists a $J' \in \Big[J^* , J_{max}\Big]$, where $J^*=\frac{8 \pi \sqrt{M} }{ \delta \sqrt{( 1 - \delta) \widetilde{\varepsilon} } }$, for which the  inequality in step (3) of Algorithm~\ref{alg:multiplicativeamplitudeestimation} is satisfied with probability at least $1-10\delta/T$.\footnote{ Note that $J^*\leq J_{max}$ follows immediately by using the lower bound $\widetilde{\varepsilon}\geq 1/QT^2$.} In order to see this, first observe that
    \begin{align}
    \begin{aligned}
    \label{eq:lemmanoste3}
    \frac{2 \sqrt{2} \pi \sqrt{ (1 - \delta) \varepsilon'}}{J\sqrt{M}} + \frac{\pi^2}{J^2} 
       \leq \frac{4 \pi \sqrt{ (1 - \delta) \widetilde{\varepsilon}}}{J\sqrt{M}} + \frac{\pi^2}{J^2} 
    \leq \frac{ \delta  (1 - \delta) \widetilde{\varepsilon }}{2 M} + \frac{\delta^2 (1 - \delta) \widetilde{\varepsilon} }{64 M} \leq  \frac{\delta  (1 - \delta) \widetilde{\varepsilon}}{M},  
        \end{aligned}
    \end{align}
  where the second inequality used $J\geq J^*$ and the remaining inequalities are straightforward. Using Eq.~\eqref{eq:bound_noinstance} and Eq.~\eqref{eq:lemmanoste3}, we have $|\varepsilon'-\widetilde{\varepsilon}|\leq \delta(1-\delta)\widetilde{\varepsilon}$.
    Moreover, using $|\widetilde{\varepsilon}-\varepsilon'|\leq \delta(1-\delta)\widetilde{\varepsilon}$, we can further upper bound Eq.~\eqref{eq:lemmanoste3} by $\delta \varepsilon'/M$, which implies step (3) of Algorithm~\ref{alg:multiplicativeamplitudeestimation} is satisfied, in which case the algorithm would have output `yes' with probability $\geq 1-10\delta/T$. Hence, by the contrapositive argument, if Algorithm~\ref{alg:multiplicativeamplitudeestimation} outputs `no' with probability at least $1-10 \delta/T$, then we have $\widetilde{\varepsilon} < 1/QT^2$. 
    
   Finally, we bound the total number of queries made to $U$ and $U^{-1}$ in Algorithm~\ref{alg:multiplicativeamplitudeestimation}. Given that~$J$ is doubled in every round and $J\leq J_{max} = O( \sqrt{MQT^2} / \delta )$, the total number of queries is
    \begin{align}
    J_{max} + \frac{J_{max} }{2} + \frac{J_{max} }{4} + \cdots +\Big\lceil\frac{2\pi \sqrt{M}}{\delta}\Big\rceil< 2 J_{max} = O\Bigg(\sqrt{MQT^2} / \delta \Bigg) = O(\sqrt{M}Q^{3/2}T^3)
    \end{align}
    using $\delta=O(1/QT^2)$ (for simplicity, we assume that all these terms are powers of $2$). This concludes the proof of the lemma.        \end{proof}%

We now describe our quantum boosting algorithm. 

     \begin{algorithm}[H]
   		\caption{Quantum boosting algorithm }
        \label{alg:quantumAdaBoost}
		\textbf{Input:} Weak quantum learner $\mathcal{A}$ with time complexity $Q$, a training sample $S =\{(x_i,c(x_i))\}_{i\in [M]}$, where $x_i$ is sampled from an unknonwn distribution $\mathcal{D}$.
		
		\vspace{5pt}
		\textbf{Initialize:} Let $\widetilde{D}^1 = D^1$ be the uniform distribution on $S$. Let $h_0$ be the constant function,\footnotemark{} $T = O((\log M) / \gamma^2 )$ and $\delta = 1/(10 Q T^2)$. $\eps'_0 = 1/2$.

		\vspace{5pt}
		\begin{algorithmic}[1]
		 \For{t = 1 to $T$ (assume quantum query access to $h_1,\ldots,h_{t-1}$ and knowledge of $\varepsilon'_1,\ldots,\varepsilon'_{t-1}$.) }
         \vspace{1mm}
		 \State Prepare $Q+1$ many copies of $\ket{\psi_1}=\frac{1}{\sqrt{M}}\sum_{x\in S} \ket{x,c(x),\widetilde{D}^1_x}$. Let $\ket{\Phi_1}=\ket{\psi_1}$.
		 
		 \vspace{3pt}
		  \hspace{-27pt}\underline{\textsf{\textbf{Phase (1): Obtaining hypothesis $h_t$}}}
		  
		  \vspace{5pt}
		 \State Using quantum queries to $\{h_1,\ldots,h_{t-1}\}$ and knowledge of $\{\varepsilon'_1,\ldots,\varepsilon'_{t-1}\}$, prepare the state 
		 $$
		 \ket{\Phi_3}=\Big(\frac{1}{\sqrt{M}}\sum_{x\in S} \ket{x,c(x),\widetilde{D}^t_x}\Big).
		 $$
		 \State Apply amplitude amplification to prepare $\ket{\Phi_6} =\Big(\sum_{x\in S} \sqrt{\widetilde{D}^t_x}\ket{x,c(x)}+\ket{\chi_t}\Big)$. 
		 \vspace{1mm}
		 \State Pass $\ket{\Phi_6}^{\otimes Q}$ to the quantum learner $\mathcal{A}$ to obtain a hypothesis $h_t$. 
		
		 \vspace{5pt}
		  \hspace{-27pt}\underline{\textsf{ \textbf{Phase (2): Estimating weighted errors $\widetilde{\varepsilon}_t$} }}
		 
		 \vspace{5pt}
		 
		 \State Using quantum queries to $h_t$, prepare $\ket{\psi_5}=\frac{1}{\sqrt{M}}\sum_{x\in S} \ket{x,c(x),\widetilde{D}^t_x \cdot  [h_t(x)\neq c(x)]}$. 
		 \vspace{1mm}
		 \State
		 Let $\widetilde{\varepsilon}_t = \Pr_{x\sim \widetilde{D}^t} [h_t(x)\neq c(x)]$. Prepare $\ket{\psi_6}= \sqrt{1 - \widetilde{\varepsilon}_t/M} \ket{\phi_0} \ket{0} + \sqrt{\widetilde{\varepsilon}_t/M} \ket{\phi_1} \ket{1}$. 
		 
		 \vspace{5pt}
		 
		 \State Invoke subroutine~\ref{alg:multiplicativeamplitudeestimation}
		 to estimate $\widetilde{\varepsilon}_t$ with $\varepsilon'_t$.
		 
		 \vspace{5pt}
		 \hspace{-27pt}\underline{\textsf{ \textbf{Phase (3): Updating distributions} }}
         \parState{%
         \textbf{If} {subroutine~\ref{alg:multiplicativeamplitudeestimation} outputs `yes': let $Z_t = 2 \sqrt{\varepsilon'_t (1 - \varepsilon'_t)}$, $ \alpha'_t = \ln \Big(\sqrt{( 1 - \varepsilon'_t)/ \varepsilon'_t } \Big)$ and update  $\widetilde{D}^{t}_x$:}
         }
    	  \begin{align}
    	  \label{eq:QAB_distribution_classicalupdate_a}
          \widetilde{D}^{t+1}_x &= \frac{\widetilde{D}^{t}_x}{(1 + 2 \delta ) Z_{t}}
          \times \begin{cases} 
          e^{-\alpha'_t} & \text{ if } h_{t}(x) = c(x) \\
          e^{\alpha'_t} & \text{ otherwise }.
          \end{cases}
          \end{align}
        \parState{%
        \textbf{If} {subroutine~\ref{alg:multiplicativeamplitudeestimation} outputs `no': let  $Z_t = \Big( 2 \sqrt{QT^2 -1} \Big)/ (QT^2) $, $ \alpha'_t = \ln \Big(\sqrt{ QT^2 - 1 } \Big)$ and update~$\widetilde{D}^{t}_x$:}
        }
		  \begin{align}
		  \label{eq:QAB_distribution_classicalupdate_b}
          \widetilde{D}^{t+1}_x &= \frac{\widetilde{D}^{t}_x}{(1 + 2/(QT^2)) Z_{t}}
          \times \begin{cases} 
          (2 - 1/(Q T^2))e^{-\alpha'_t} & \text{ if } h_{t}(x) = c(x) \\
          (1/(QT^2))e^{\alpha'_t} & \text{ otherwise }.
          \end{cases}
           \end{align}
		\EndFor	
		\end{algorithmic}
		
		\textbf{Output:} Hypothesis $H$ defined as $H(x) = \sgn \Big( \sum_{t=1}^T \alpha'_t h_t(x) \Big)$ for all $x\in \01^n$.
	\end{algorithm}
	\footnotetext{Precisely, we let the query operation $\mathcal{O}_{h_0}$ corresponding to $h_0$ be the identity  map.}
%

 Before describing the state of the quantum boosting algorithm in every step, we make a couple of remarks.
 We use the notation $\widetilde{D}^t_x$ in the quantum boosting algorithm because $\{\widetilde{D}^t_x\}_x$ is not a true distribution since it satisfies $\sum_{x \in S} \widetilde{D}^t_x \leq 1$ (this is also the reason for incorporating the state $\ket{\chi_t}$ in step~$(4)$). In addition to $\widetilde{D}^t_x$, we also define the \emph{true} updated distribution $D^{t+1}_x$ as follows
   \begin{align}
    \label{eq:QAB_distribution_classicalupdatetrue}
    {D}^{t+1}_x=  \frac{\widetilde{D}^t_x}{Z_t} \times \begin{cases} 
      e^{-\alpha'_t} & \text{ if } h_t(x) = c(x) \\
      e^{\alpha'_t} & \text{ otherwise },
   \end{cases}
   \end{align}
where  $\alpha'_t = \ln \Big(\sqrt{( 1 - \varepsilon'_t)/ \varepsilon'_t } \Big)$ and  $\varepsilon'_t,\widetilde{\varepsilon}_t$ are defined in step $(8)$ of Algorithm~\ref{alg:quantumAdaBoost}, and
\begin{align}
\label{eq:defnofZt}
\begin{aligned}
Z_t = \sum_{i=1}^M \widetilde{D}_t(x_i) \exp{\big(-\alpha_t^\prime h_t(x_i)c(x_i) \big) }
     &= \sum_{i:h_t(x_i) = c(x_i)} \widetilde{D}_t(x_i) \cdot  e^{-\alpha_t^\prime } + \sum_{i:h_t(x_i) \neq c(x_i)} \widetilde{D}_t(x_i)\cdot  e^{\alpha_t^\prime}\\
  &= (1 - \widetilde{\varepsilon}_t) \cdot  e^{-\alpha'_t} + \widetilde{\varepsilon}_t \cdot  e^{\alpha_t^\prime}.
  \end{aligned}
  \end{align}
Note that $\sum_{x \in S} D_x^{t+1}= 1$ and observe that our quantum boosting algorithm cannot make the distribution update  Eq.~\eqref{eq:QAB_distribution_classicalupdatetrue} since the value of $\widetilde{\varepsilon}_t$ in Eq.~\eqref{eq:defnofZt} is unknown to the quantum algorithm. Let~$\varepsilon_t$ be  the weighted error given by $\varepsilon_t = \Pr_{x\sim D^t} [h_t(x)\neq c(x)]$ corresponding to the true distribution~$\{D^t_x\}_x$.

 This is one of the main differences between standard AdaBoost and our quantum boosting Algorithm~\ref{alg:quantumAdaBoost}. In standard AdaBoost, one assumes that the $\varepsilon$s can be computed exactly by spending time $O(M)$. However, a quantum algorithm can only approximate the $\varepsilon$s in time $O(\sqrt{M})$. Hence the distribution update from $D^t\rightarrow D^{t+1}$ in classical AdaBoost might result in a sub-normalized distribution $\widetilde{D}^{t+1}$ in the quantum case. Moreover even if a learner could produce a hypothesis $h_t$ that is $(1/2+\gamma)$-close to the target concept $c$ ``under" $\widetilde{D}^{t+1}$, it is not clear if the final hypothesis $H$ has small training error. In order to overcome this, we split the distribution update step into two cases (steps $(9,10)$ in Algorithm~\ref{alg:quantumAdaBoost}) depending on whether $\varepsilon$ is ``large" or ``small". Using the structure of the distribution update we show that the resulting hypothesis~$H$ has small training error. Before we proceed with the proof of the main theorem, we state the following properties about the distributions $\widetilde{D}s$.
\begin{claim}
\label{claim:propertyofD_a}
Let $t\geq 1$, $\widetilde{D}^t:\01^n\rightarrow [0,1]$ be defined as in Eq.~\eqref{eq:QAB_distribution_classicalupdate_a}, \eqref{eq:QAB_distribution_classicalupdate_b}. Then $\sum_{x \in S} \widetilde{D}^t_x  \in [1-30 \delta,1]$. 
\end{claim}

\begin{claim}
\label{claim:distancefromtrueeps}
Let $t\geq 1$, $\widetilde{\varepsilon}_t = \Pr_{x\sim \widetilde{D}^t} [h_t(x)\neq c(x)]$ be the weighted error corresponding to the approximate distribution $\widetilde{D}^t$ and $\varepsilon_t = \Pr_{x\sim D^t} [h_t(x)\neq c(x)]$ correspond to the true distribution $D^t$. Then $\vert  \widetilde{\varepsilon}_t - \varepsilon_t \vert \leq 50 \delta$. 
\end{claim}

We prove these claims later. Observe that our quantum boosting algorithm will run on the sub-normalized distribution $\widetilde{D}^t$ instead of the ideal distribution $D^t$,  since we do not have knowledge of $\widetilde{\varepsilon}_t$ in Eq.~\eqref{eq:defnofZt} and instead have an estimate $\varepsilon'_t$ of $\widetilde{\varepsilon}_t$. 
We now describe the unitary operation that updates $\widetilde{D}_1$ (in step $(2)$) to $\widetilde{D}_t$ (in step $(3)$). For  $t\in \{1,\ldots,T\}$, let $\G_t$ be the quantum circuit that makes the distribution update from $\widetilde{D}^1\rightarrow \widetilde{D}^t$. Given access to $h_1,\ldots,h_{t-1}, c:\01^n\rightarrow \{-1,1\}$ and knowledge of $\varepsilon'_1,\ldots,\varepsilon'_{t-1}$, define~$\G_t$ as the map:
    \begin{align}
    \label{eq:QAB_distribution_quantumupdateC}
    \begin{split}
    \mathcal{G}_t &: \frac{1}{\sqrt{M}} \sum_{x\in S} \ket{x,c(x)} \otimes \ket{\widetilde{D}^1_x} \otimes \ket{[h_1(x) \neq c(x)], \ldots, [h_{t-1}(x) \neq c(x)]} 
    \\
    &\quad
    \rightarrow  \frac{1}{\sqrt{M}} \sum_{x\in S} \ket{x,c(x)} \otimes \ket{\widetilde{D}^{t}_x} \otimes \ket{[h_1(x) \neq c(x)], \ldots, [h_{t-1}(x) \neq c(x)]}.
    \end{split}
    \end{align}
    
    \paragraph{Details of Algorithm~\ref{alg:quantumAdaBoost}.} We are now ready to describe our quantum boosting algorithm in more details. In the $t$th step, we have quantum query access to the hypotheses $\{h_1,\ldots,h_{t-1}\}$ and knowledge of the approximate weighted errors  $\{\varepsilon'_1,\ldots,\varepsilon'_{t-1}\}$. 
      Let $U_1, U_2$ be unitaries that satisfy $U_1:\ket{0}\rightarrow \ket{\psi_1}$ and $U_2:\ket{0}\rightarrow \ket{\Phi_1}$, where
    $$
    \ket{\psi_1} = \Bigg( \frac{1}{\sqrt{M}} \sum_{x\in S} \ket{x,c(x)} \otimes \ket{\widetilde{D}^1_x} \otimes \ket{0}^{\otimes {t+1}} \Bigg),\quad  \ket{\Phi_1}= \Bigg( \frac{1}{\sqrt{M}} \sum_{x\in S} \ket{x , c(x)} \otimes \ket{ \widetilde{D}^1_x} \otimes \ket{0}^{\otimes t} \Bigg).
    $$
Recall that $\widetilde{D}^1$ is the uniform distribution over the training set $S$. We assume that theoretically one could use a quantum random access memory (QRAM) to prepare the state $\ket{\psi_0}=\frac{1}{\sqrt{M}}\sum_{x\in S}\ket{x,c(x)}$ in time $O(\log M)$. By ``use a QRAM to prepare" we mean, we can perform the operation that takes the training set $S=\{(x_i,c(x_i))\}_{i \in [M]}$ stored in a classical data structure and prepares the \emph{uniform} quantum state $\ket{\psi_0}$ in time $O(\log M)$.  Given the QRAM assumption has been controversial and seems strong in quantum machine learning, we make a couple of remarks: (i) our quantum  boosting algorithm \emph{only} requires a QRAM to prepare the \emph{uniform} superposition over classical data~$S$ at the start of each iteration. Also, our quantum algorithm does not use QRAM as an oracle for Grover-like algorithms, so the negative results of~\cite{arunachalam:qram} do not apply to our~algorithm; (ii) we use the QRAM at the start of $T=O(\log M)$ iteration of our algorithm to prepare $\ket{\psi_0}$, so even if the quantum time complexity of preparing $\ket{\psi_0}$ is $O(\sqrt{M})$, then our complexity increases by an \emph{additive} $O(\sqrt{M}\log M)$ term and we  still do not lose our quantum speedup; (iii) of course if QRAM is infeasible then we can also assume that a quantum learner has access to uniform  quantum examples $\ket{\psi_0}$ or has \emph{quantum query access} to the training examples in $S$ (i.e., can perform the map $\ket{x,b}\rightarrow \ket{x,b\cdot c(x)}$ for $x\in S$). and in both  cases we do not need a~QRAM.

    We now describe the quantum boosting algorithm. The algorithm begins by first preparing $ U_1\ket{0}\otimes U_2^{\otimes Q} \ket{0}=\ket{\psi_1} \otimes \ket{\Phi_1}^{\otimes Q}$. We then apply $(Q+1)(t-1)$ quantum queries to the oracles $\{O_{h_1}, \ldots, O_{h_{t-1}}\}$ to obtain\footnote{ To be precise, we obtain the $t$th indicator function $[h_t(x) \neq c(x)]$ as follows: first use $  O_{h_t}$ to perform the map $\frac{1}{\sqrt{M}} \sum_{x\in S} \ket{x,c(x)} \ket{\widetilde{D}^1_x} \ket{0} = \frac{1}{\sqrt{M}} \sum_{x\in S} \ket{x , c(x)} \ket{ \widetilde{D}^1_x} \ket{h_t(x) } $, next apply the CNOT gate between the second and fourth register to produce $ \frac{1}{\sqrt{M}} \sum_{x\in S} \ket{x , c(x)} \ket{ \widetilde{D}^1_x} \ket{h_t(x) \cdot c(x) }$}
    \begin{dmath}
    \label{eq:QAB_algorithm_(t-1)query}  
      \ket{\psi_2} \otimes \ket{\Phi_2}^{\otimes Q} = \Bigg( \frac{1}{\sqrt{M}} \sum_{x\in S} \ket{x,c(x)} \otimes \ket{\widetilde{D}^1_x} \otimes \ket{[h_1(x) \neq c(x)], \ldots ,[h_{t-1}(x) \neq c(x)]}\ket{0}^2 \Bigg) \otimes 
		   \Bigg( \frac{1}{\sqrt{M}} \sum_{x\in S} \ket{x , c(x)} \ket{ \widetilde{D}^1_x} \ket{[h_1(x) \neq c(x)], \ldots ,[h_{t-1}(x) \neq c(x)], 0}  \Bigg)^{\otimes Q} .
    \end{dmath}
    We then apply the unitary $\G_{t}$ on $\ket{\psi_2}$ and each of the $Q$ copies of $\ket{\Phi_2}$ to update $\widetilde{D}^1$. The resulting state is
    \begin{dmath}
    \label{eq:QAB_algorithm_distributionupdate}
	  \ket{\psi_3} \otimes \ket{\Phi_3}^{\otimes Q} = \Bigg( \frac{1}{\sqrt{M}} \sum_{x\in S} \ket{x,c(x)} \otimes \ket{\widetilde{D}^t_x} \otimes \ket{[h_1(x) \neq c(x)], \ldots ,[h_{t-1}(x) \neq c(x)]}\ket{0}^2 \Bigg) \otimes \Bigg( \frac{1}{\sqrt{M}} \sum_{x\in S} \ket{x , c(x)} \ket{ \widetilde{D}^t_x} \ket{[h_1(x) \neq c(x)], \ldots, [h_{t-1}(x) \neq c(x)],0}  \Bigg)^{\otimes Q} .
	\end{dmath}
    We now break down the algorithm into two steps, the first phase uses $\ket{\Phi_3}^{\otimes Q}$ to obtain $h_t$ and the second phase uses $h_t$ and $\ket{\psi_3}$ to compute $\varepsilon'_t$. 
    
    \paragraph{Phase (1): Obtaining hypothesis $h_t$.}  Let $V : \ket{p} \ket{0} \ket{0} \rightarrow \ket{p}  \Big( \sqrt{1-p} \ket{0} + \sqrt{p} \ket{1} \Big) \ket{ \sin^{-1} \big(\sqrt{p}\big) }$. For each of the $Q$ copies of $\ket{\Phi_3}$, append an auxiliary $\ket{0}$ and apply $V$ to obtain
        \begin{align}
    \label{eq:QAB_algorithm_controlledrefNprime}
    \ket{\Phi_4} = \frac{1}{\sqrt{M}} \sum_{x\in S} \ket{x , c(x)} \ket{ \widetilde{D}^t_x} \ket{[h_1(x) \neq c(x)], \ldots, [h_{t-1}(x) \neq c(x)]} \Big( \sqrt{\widetilde{D}^t_x} \ket{0} + \sqrt{1 - \widetilde{D}^t_x} \ket{1} \Big).
    \end{align}
    Let $W_t$ be the unitary that performs $W_t : \ket{\Phi_1} \rightarrow \ket{\Phi_4}$ and $\widetilde{W}_t = W_t U_2$ be the map $\widetilde{W}_t : \ket{0} \rightarrow \ket{\Phi_4}$. Let $Y_t$ be the unitary that uses an expected $O(\sqrt{M} \log T)$ invocations of $\widetilde{W}_t$ and $\widetilde{W}_t^{-1}$ to perform amplitude amplification (in Theorem~\ref{thm:amp_amplification}) on the state $\ket{\Phi_4}$ and with probability at least $1-O(1/T)$ produce the state $\ket{\Phi_5}$
    \begin{align}
    \label{eq:QAB_algorithm_amplitudeamplify}
    \ket{\Phi_5} =  Y_t \ket{\Phi_4} = \sum_{x\in S} \sqrt{\widetilde{D}^t_x} \ket{x , c(x)} \ket{\widetilde{D}^t_x} \ket{[h_1(x) \neq c(x)], \ldots, [h_{t-1}(x) \neq c(x)],0} +\ket{\Psi},
    \end{align}
    where $\ket{\Psi}$ is orthogonal to the first register. Observe that without $\ket{\Psi}$, the state $\ket{\Phi_5}$ is no longer a quantum state because $\{\widetilde{D}^t_x\}_x$ is a sub-normalized distribution and note that the complexity of amplitude amplification is $\Theta\Big(\sqrt{M\big/\sum_{x \in S} \widetilde{D}^t_x}\Big)=\Theta(\sqrt{M})$ using Claim\ref{claim:propertyofD_a}. Moreover, using  Claim~\ref{claim:propertyofD_a}, we have
    \begin{align}
    \label{eq:propertyofPsi}
     \|\ket{\Psi} \|\leq 1-\sum_{x \in S}\widetilde{D}^t_x\leq 30 \delta. \end{align}
Observe that if we had run the boosting algorithm with the ideal distribution $D^t$, we would have obtained the state
\begin{align}
    \label{eq:Phiprime}
    \ket{\Phi'_5}  = \sum_{x\in S} \sqrt{{D}^t_x} \ket{x , c(x)} \ket{{D}^t_x} \ket{[h_1(x) \neq c(x)], \ldots, [h_{t-1}(x) \neq c(x)],0},
\end{align}
instead of  $\ket{\Phi_5}$. We now uncompute the auxiliary registers $\ket{\widetilde{D}^t_x} \ket{[h_1(x) \neq c(x)], \ldots, [h_{t-1}(x) \neq c(x)],0}$ in $\ket{\Phi_5}$ as follows: let $\G_t^{-1}$ be the unitary which maps $\ket{\widetilde{D}^t_x}\rightarrow \ket{\widetilde{D}^1_x}$ and let $O_{h_1},\ldots,O_{h_{t-1}}$ be the query operations that uncompute the $\{[h_i(x) \neq c(x) ] \}_{i \in [t-1]}$ registers. Applying $\G_t^{-1}$ on the \emph{actual} state~$\ket{\Phi_5}$ (instead of the \emph{ideal} state $\ket{\Phi_5'}$) gives
$$
\G_t^{-1}\ket{\Phi_5}=   \sum_{x\in S} \sqrt{\widetilde{D}^t_x} \ket{x , c(x)} \ket{0} \ket{[h_1(x) \neq c(x)], \ldots, [h_{t-1}(x) \neq c(x)],0}+\G_{t}^{-1}\ket{\Psi},
$$
  and then performing $O_{h_1}^{-1},\ldots,O_{h_{t-1}}^{-1}$ gives us 
    \begin{align}
        \label{eq:phi6}
    \ket{\Phi_6}= \sum_{x\in S} \sqrt{\widetilde{D}^t_x} \ket{x , c(x)} \ket{0} \ket{0}^t+O_{h_{t-1}}\cdots O_{h_1}\cdot \G_{t}^{-1}\ket{\Psi}.
    \end{align}    
By performing the operations $\G_t^{-1}$, $O_{h_1}^{-1},\ldots,O_{h_{t-1}}^{-1}$ on the \emph{ideal} state $\ket{\Phi'_5}$, we would have 
    \begin{align}
        \label{eq:phi6'}
        \ket{\Phi'_6}= \sum_{x\in S} \sqrt{{D}^t_x} \ket{x , c(x)} \ket{0} \ket{0}^t.
    \end{align}
   Ideally, our goal would be to pass $Q$ copies of $\ket{\Phi'_6}$ to a quantum learner in order to obtain a hypothesis $h_t$. Although, we do not have access to $\ket{\Phi_6'}$, we continue our quantum boosting algorithm by passing $Q$ copies of $\ket{\Phi_6}$ to a quantum learner (instead of $Q$ copies of $\ket{\Phi_6'}$). A priori, it is not clear what will be the output of the quantum learner on input $\ket{\Phi_6}^{\otimes Q}$. In order to understand this, we first show that $\ket{\Phi_6}$ and $\ket{\Phi_6'}$ are close. Using this, it is not hard to see that a quantum learner would \emph{behave similarly} when given $\ket{\Phi_6}^{\otimes Q}$ instead of $\ket{\Phi_6'}^{\otimes Q}$. In order to formalize this, we first state the following claim which we prove later.
    \begin{claim}
    \label{claim:distancebetweenstates_a}
    Let $\ket{\Phi_6}$ and $\ket{\Phi'_6}$ be as defined in Eq.~\eqref{eq:phi6},~\eqref{eq:phi6'}. Then we have $|\langle \Phi_6 \vert \Phi'_6 \rangle| \geq 1-50 \delta$.
    \end{claim}
Recall that $\ket{\Phi_6'}$ is the ideal state that satisfies the following: suppose $Q$ copies of $\ket{\Phi'_6}$ are given to a weak quantum learner, then with probability at least $1-1/T$, the learner outputs a weak hypothesis~$h_t$. We now show that the same learner, when fed $Q$ copies of $\ket{\Phi_6}$ (instead of $\ket{\Phi'_6}$) will output~$h_t$ with probability at least $1-17/T$. In order to see this, let $p'=\mathsf{Pr}[ \A$ outputs $h_t$ given $\ket{\Phi'_6}^{\otimes Q}] \geq 1 - 1/T$ and $p=\mathsf{Pr}[ \A$ outputs $h_t$ given $\ket{\Phi_6}^{\otimes Q}]$. Let $\mathcal{H}$ be the hypothesis class and suppose $\{\mathsf{E}_{h_t}\}_{h_t \in \mathcal{H}}$ (satisfying $\sum_{h_t \in \mathcal{H}} \mathsf{E}_{h_t} = \Id$) is the final POVM performed by $\A$. Then we have the following,
    \begin{align}  
    \label{eq:boundingprobofp-p'1_a}
    \begin{aligned}
	\big\lvert p'-p\big\rvert &=\Big\lvert \Tr \big( \mathsf{E}_{h_t} \vert \Phi'_6 \rangle \langle \Phi'_6 \vert^{\otimes Q} \big) - \Tr \big( \mathsf{E}_{h_t}  \vert \Phi_6 \rangle \langle \Phi_6 \vert^{\otimes Q}  \big) \Big\rvert\\
		&\leq \sum_{h_t \in \mathcal{H}} \Big\lvert \Tr \big( \mathsf{E}_{h_t} \vert \Phi'_6 \rangle \langle \Phi'_6 \vert^{\otimes Q} \big) - \Tr \big( \mathsf{E}_{h_t}  \vert \Phi_6 \rangle \langle \Phi_6 \vert^{\otimes Q}  \big) \Big\rvert \\
		&\leq 2 \norm{ (|\Phi'_6 \rangle \langle \Phi'_6 |)^{\otimes Q} - ( |\Phi_6 \rangle \langle \Phi_6 | )^{\otimes Q}  }_1
		\\
		&=4(1-\langle \Phi_6\vert \Phi'_6\rangle^{2Q})^{1/2} 
		\\
		&\leq 4(1-(1-50\delta )^{2Q})^{1/2} \leq 4(1-(1-50 Q \delta)^2)^{1/2}\leq 16/T,
	\end{aligned}
	\end{align}
	where we have used the definition of trace distance in the second equality, Claim~\ref{claim:distancebetweenstates_a} in the third inequality, Bernoulli's inequality $(1-x)^t\geq 1-xt$  (for $x\leq 1$ and $t\geq 0$) in the penultimate inequality and $\delta=1/(10 QT^2)$ in the final inequality. Additionally, the second inequality follows from the definition of the trace~distance between quantum states
    $$
    \norm{\rho - \sigma}_1 = \max_{\{\mathsf{E}_m\}} \frac{1}{2} \sum_m \vert \Tr(E_m (\rho - \sigma))  \vert.
    $$
    In particular, suppose we have a weak learner $\A$ that outputs $h_t$ with probability at least $p'$, then for every $t \in [T]$, we have 
	\begin{align}  
	p=\mathsf{Pr}[ \A \text{ outputs }h_t \text{ given }\ket{\Phi_6}^{\otimes Q}] \geq p' - 16/T \geq 1 - 1/T - 16/T = 1 - 17/T.
	\end{align}
        Hence, on passing $\ket{\Phi_6}^{\otimes Q}$ to a quantum learner $\mathcal{A}$, it outputs a weak hypothesis $h_t$ with  probability at least $ 1 - 17/T$ using Eq.~\eqref{eq:boundingprobofp-p'1_a}. We assume that the output $h_t$ is presented in terms of an oracle~$O_{h_t}$ (which on query $\ket{x,b}$ outputs $\ket{x,b\cdot h_t(x)}$ for all $b\in \pmset{},x\in \01^n$). 
    
    \paragraph{Phase (2): Computing $\varepsilon'_t$.}   Using $O_{h_t}$ produced in Phase (1), we now perform the query operation~$O_{h_t}$ on $\ket{\psi_3}$ (defined in Eq.~\eqref{eq:QAB_algorithm_distributionupdate}) and obtain
	\begin{align}
	\label{eq:QAB_algorithm_(t-th)query}
	  \ket{\psi_4}=O_{h_t}\ket{\psi_3}= \frac{1}{\sqrt{M}} \sum_{x\in S} \ket{x,c(x)} \otimes \ket{\widetilde{D}_x^t} \otimes \ket{[h_1(x) \neq c(x)], \ldots, [h_{t}(x) \neq c(x)],0} .
	\end{align}
    Using arithmetic operations, one can additionally produce the following state
    \begin{align}
    \label{eq:QAB_algorithm_arithmeticprocess}
	  \ket{\psi_5} = \frac{1}{\sqrt{M}} \sum_{x\in S} \ket{x,c(x)} \otimes \ket{\widetilde{D}_x^t \cdot [h_{t}(x) \neq c(x)] } \otimes \ket{[h_1(x) \neq c(x)], \ldots, [h_{t}(x) \neq c(x)],0}.
    \end{align}
    We now apply the controlled reflection operator $V: \ket{p} \ket{0} \ket{0} \rightarrow \ket{p}  \Big( \sqrt{1-p} \ket{0} + \sqrt{p} \ket{1} \Big) \ket{ \sin^{-1} \big(\sqrt{p}\big) }$ (where we store $\sin^{-1}(\cdot)$ up to $n$ bits of accuracy) to~$\ket{\psi_5}$ to obtain
    $$
      \ket{\psi_6}= \frac{1}{\sqrt{M}} \sum_{x\in S} \ket{x,c(x)} \otimes \ket{ \beta_x^t } \otimes \ket{[h_1(x) \neq c(x)], \ldots, [h_{t}(x) \neq c(x)]} \otimes \Big( \sqrt{1 - \beta^t_x} \ket{0} + \sqrt{\beta^t_x} \ket{1} \Big) \otimes \ket{\sin^{-1} \big(\sqrt{ \beta_x^t }\big)},
    $$
    where $\beta^t_x = \widetilde{D}^t_x [h_{t}(x) \neq c(x)]$. We  can rewrite the above equation as
      \begin{align}
    \label{eq:QAB_algorithm_controlreflection}
    \ket{\psi_6} &= \sqrt{ \widetilde{\varepsilon}_t/ M } \ket{\phi_1} \ket{1}+ \sqrt{1 - \widetilde{\varepsilon}_t/M } \ket{\phi_0} \ket{0} ,
    \end{align}
     where  $\widetilde{\varepsilon}_t = \sum_{x\in S } \beta_x^t = \sum_{x\in S }\widetilde{D}^t_x [h_{t}(x) \neq c(x)]$ and $\ket{\phi_0}$, $\ket{\phi_1}$ are defined as
    \begin{subequations} \label{eq:QAB_algorithm_normalizedstates}
	\begin{align} 
	  \ket{\phi_0} &= \frac{1}{\sqrt{ M }}  \sum_{x\in S}  \frac{\sqrt{1- \beta_x^t} }{\sqrt{1- \widetilde{\varepsilon}_t/ M } } \ket{x,c(x)} \otimes \ket{ \beta_x^t } \otimes \ket{[h_1(x) \neq c(x)], \ldots, [h_{t}(x) \neq c(x)]} \otimes  \ket{\sin^{-1} \big(\sqrt{ \beta_x^t }\big)}, 
		\\
	  \ket{\phi_1} &= \frac{1}{\sqrt{ M }}  \sum_{x\in S}  \frac{\sqrt{\beta_x^t} }{\sqrt{ \widetilde{\varepsilon}_t/ M } } \ket{x,c(x)} \otimes \ket{ \beta_x^t } \otimes \ket{[h_1(x) \neq c(x)], \ldots, [h_{t}(x) \neq c(x)]} \otimes  \ket{\sin^{-1} \big(\sqrt{ \beta_x^t }\big)}  . 
	\end{align}
	\end{subequations}
    Let $F_t$ be the unitary given by the map $F_t : \ket{\psi_1} \rightarrow \ket{\psi_6}$ and $\widetilde{F}_t = F_t U_1$ be the map $\widetilde{F}_t : \ket{0} \rightarrow \ket{\psi_6} $. Let $P_t$ be the unitary that implements amplitude estimation using $ J_t $ invocations of $\widetilde{F}_t$ and $\widetilde{F}_t^{-1}$.  In order to approximate $\widetilde{\varepsilon}_t$, we run Algorithm~\ref{alg:multiplicativeamplitudeestimation} on the state $\ket{\psi_6}$ assuming  unitary access to $\widetilde{F}_t$. Depending on the output $\varepsilon'_t$ of Algorithm~\ref{alg:multiplicativeamplitudeestimation}, we  compute $\alpha'_t=\frac{1}{2}\ln \Big(\frac{1-\varepsilon'_t}{ \varepsilon'_t}\Big)$. Using $\varepsilon'_t$ and $\alpha'_t$, we now update the distribution from $\widetilde{D}^t$ to $\widetilde{D}^{t+1}$,

\subsubsection{Proof of claims}
\label{sec:claimsboosting}
In this section, we state and prove a few claims from the previous section. We restate these claims for convenience of the reader. Additionally, we will crucially use the following relations multiple times in this section: for every $t\geq 1$, let $\widetilde{D}^t$ be the sub-normalized distribution defined in Algorithm~\ref{alg:quantumAdaBoost} (in particular  Eq.~\eqref{eq:QAB_distribution_classicalupdate_a}) when Algorithm~\ref{alg:multiplicativeamplitudeestimation} outputs `yes',  then
\begin{align}
\label{eq:defnofZtrepeat}
\begin{aligned}
 \sum_{x \in S} \widetilde{D}^t_x \exp(-\alpha'_t c(x) h_t(x))&=\sum_{i:h_t(x_i) = c(x_i)} \widetilde{D}^t(x_i) \cdot  e^{-\alpha_t^\prime } + \sum_{i:h_t(x_i) \neq c(x_i)} \widetilde{D}^t(x_i)\cdot  e^{\alpha_t^\prime}\\
  &= (1 - \widetilde{\varepsilon}_t) \cdot  e^{-\alpha'_t} + \widetilde{\varepsilon}_t \cdot  e^{\alpha_t^\prime},
  \end{aligned}
  \end{align}
  where $\widetilde{\varepsilon}_t = \Pr_{x\sim \widetilde{D}^t} [h_t(x)\neq c(x)]$.
 Also let $\widetilde{D}^t$ be the  sub-normalized distribution defined in Algorithm~\ref{alg:quantumAdaBoost} (in particular  Eq.~\eqref{eq:QAB_distribution_classicalupdate_b}) when Algorithm~\ref{alg:multiplicativeamplitudeestimation} outputs `no', then 
 \begin{align}
\label{eq:defnofZtrepeat_b}
\begin{aligned}
 \sum_{x \in S} \widetilde{D}^t_x \exp(-\alpha'_t c(x) h_t(x)+ \kappa_{[h_t(x) \neq c(x)]} )&=\sum_{i:h_t(x_i) = c(x_i)} \widetilde{D}^t(x_i) \cdot  e^{-\alpha_t^\prime + \kappa_0} + \sum_{i:h_t(x_i) \neq c(x_i)} \widetilde{D}^t(x_i)\cdot  e^{\alpha_t^\prime + \kappa_1}\\
  &= (1 - \widetilde{\varepsilon}_t) \cdot  e^{-\alpha'_t + \kappa_0 } + \widetilde{\varepsilon}_t \cdot  e^{\alpha_t^\prime + \kappa_1}.
  \end{aligned}
  \end{align}

\begin{customthm}{4.3}
\label{claim:propertyofD}
Let $t\geq 1$, $\widetilde{D}^t:\01^n\rightarrow [0,1]$ be as defined in Eq.~\eqref{eq:QAB_distribution_classicalupdate_a}, \eqref{eq:QAB_distribution_classicalupdate_b}. Then $\sum_{x \in S} \widetilde{D}^t_x  \in [1-30 \delta,1]$. 
\end{customthm}

\begin{proof} We divide the proof of the claim into two cases. Recall $\delta = 1/(10 Q T^2)$.

\textbf{Case I:} Suppose Algorithm~\ref{alg:multiplicativeamplitudeestimation} outputs `yes' in the $t$th iteration. Recall the definition of $\widetilde{D}^{t+1}$.
\begin{align}
    \widetilde{D}^{t+1}_x=  \frac{\widetilde{D}^{t}_x}{2(1+2\delta) \sqrt{\varepsilon'_t (1 - \varepsilon'_t) }} \times \begin{cases} 
      e^{-\alpha'_t} & \text{ if } h_t(x) = c(x) \\
      e^{\alpha'_t} & \text{ otherwise },
   \end{cases}
   \end{align}
  where $\alpha'_t=\frac{1}{2}\ln \Bigg(\frac{1-\varepsilon'_t}{\varepsilon'_t} \Bigg)$ and $|\widetilde{\varepsilon}_t - \varepsilon'_t|\leq \delta \varepsilon'_t$. In order to prove the lower bound, observe that
  \begin{align*}
      \sum_{x \in S} \widetilde{D}^{t+1}_x&=\frac{1}{2(1+2\delta) \sqrt{\varepsilon'_t (1 - \varepsilon'_t) }} \sum_{x\in S} \widetilde{D}^t_x \exp(-\alpha'_t c(x) h_t(x))
      \\
      &= \frac{\sum_{x \in S} \widetilde{D}^t_x \exp(-\alpha'_t c(x) h_t(x))}{(1-\widetilde{\varepsilon}_t)e^{-\alpha'_t}+\widetilde{\varepsilon}_t e^{\alpha'_t}} \cdot \frac{(1-\widetilde{\varepsilon}_t)e^{-\alpha'_t}+ \widetilde{\varepsilon}_t e^{\alpha'_t}}{2(1+2 \delta) \sqrt{\varepsilon'_t (1 - \varepsilon'_t) }}
      \\
      &=\frac{(1-\widetilde{\varepsilon}_t)e^{-\alpha'_t}+\widetilde{\varepsilon}_te^{\alpha'_t}}{2(1+2\delta) \sqrt{\varepsilon'_t (1 - \varepsilon'_t) }}\tag{using Eq.~\eqref{eq:defnofZtrepeat}}
      \\
      &=\frac{1}{2(1+2\delta)} \frac{1}{\sqrt{\varepsilon'_t (1 - \varepsilon'_t) }} \cdot \Bigg( (1-\widetilde{\varepsilon}_t)\sqrt{\frac{\varepsilon'_t}{1-\varepsilon'_t}}+ \widetilde{\varepsilon}_t \sqrt{\frac{1-\varepsilon'_t}{\varepsilon'_t}} \Bigg)\tag{using the definition of $\alpha'_t$}
      \\
      &=\frac{1}{2(1+2 \delta)} \Bigg( \frac{1-\widetilde{\varepsilon}_t}{1-\varepsilon'_t}+ \frac{\widetilde{\varepsilon}_t}{\varepsilon'_t}\Bigg),
  \end{align*}
  where the second equality used $\sum_{x\in S} \widetilde{D}^t_x \exp(-\alpha'_t c(x) h_t(x))= (1-\widetilde{\varepsilon}_t)e^{-\alpha'_t}+\widetilde{\varepsilon}_t e^{\alpha'_t}$ (which follows from Eq.~\eqref{eq:defnofZt}).  
  Since $|\widetilde{\varepsilon}_t - \varepsilon'_t|\leq \delta \varepsilon'_t$, we have
\begin{align}
    \label{eq:lowerboundeps/eps'_a}
    \frac{\widetilde{\varepsilon}_t}{\varepsilon'_t} \geq \frac{\varepsilon'_t( 1 -\delta)}{\varepsilon'_t}= 1-\delta.
\end{align}
Additionally,
\begin{align}
\label{eq:lowerbound1-eps/1-eps'_a}
    \frac{1-\widetilde{\varepsilon}_t}{1-\varepsilon'_t}\geq \frac{1-\varepsilon'_t (1+\delta) }{1-\varepsilon'_t}=1-\frac{\delta \varepsilon'_t}{1-\varepsilon'_t}\geq  1-2\delta,
\end{align}
where the second inequality uses $\varepsilon'_t\leq 2/3$ (since we assume $\varepsilon_t\leq 1/2$). Putting together Eq.~\eqref{eq:lowerboundeps/eps'_a} and Eq.~\eqref{eq:lowerbound1-eps/1-eps'_a} into the expression for $ \sum_{x\in S} \widetilde{D}^{t+1}_x$ , we get
$$
\sum_{x\in S} \widetilde{D}^{t+1}_x\geq \frac{2-3\delta}{2(1+2\delta)}\geq 1-4\delta\geq 1-30\delta.
$$
Next, we prove the upper bound. Note that 
\begin{align}
\label{eq:upperboundeps/eps'_a}
    \frac{\widetilde{\varepsilon}_t}{\varepsilon'_t} \leq \frac{\varepsilon'_t( 1 + \delta)}{\varepsilon'_t}= 1 + \delta,
\end{align}
and
\begin{align}
\label{eq:upperbound1-eps/1-eps'_a}
    \frac{1 - \widetilde{\varepsilon}_t}{1 - \varepsilon'_t} \leq \frac{1-\varepsilon'_t (1 - \delta) }{1-\varepsilon'_t}=1 + \frac{\delta \varepsilon'_t}{1-\varepsilon'_t}\leq  1 + 2\delta,
\end{align}
where the second inequality uses $\varepsilon'_t\leq 2/3$. Using Eq.~\eqref{eq:upperboundeps/eps'_a},~\eqref{eq:upperbound1-eps/1-eps'_a}, we have
\begin{align*}
\sum_{x\in S} \widetilde{D}^{t+1}_x&=\frac{1}{2(1+2\delta)} \Bigg( \frac{1-\widetilde{\varepsilon}_t}{1-\varepsilon'_t}+ \frac{\widetilde{\varepsilon}_t}{\varepsilon'_t} \Bigg)
\leq   \frac{1}{2(1+2 \delta)}\cdot \Big( ( 1 + 2 \delta ) + ( 1 + \delta) \Big) 
\leq  \frac{2 + 4 \delta}{2(1+2\delta)} = 1.
\end{align*}

\textbf{Case II:} Suppose Algorithm~\ref{alg:multiplicativeamplitudeestimation} outputs `no' in the $t$th iteration.  The distribution  $\widetilde{D}^{t+1}$ is then updated according to 
  \begin{align}
    \widetilde{D}^{t+1}_x &= \frac{\widetilde{D}^{t}_x}{(1 + 2/(QT^2)) Z_{t}}
          \times \begin{cases} 
          (2 - 1/(Q T^2))e^{-\alpha'_t} & \text{ if } h_{t}(x) = c(x) \\
          (1/(QT^2))e^{\alpha'_t} & \text{ otherwise },
          \end{cases}
   \end{align}
  where we use $\varepsilon'_t = 1/(Q T^2)$, $\alpha'_t=\ln \sqrt{ (1-\varepsilon'_t)/\varepsilon'_t}$ and $Z_t = 2 \sqrt{\varepsilon'_t (1 - \varepsilon'_t) }$. Let $\kappa_0 = \ln(2 - 1/(Q T^2))$ and $\kappa_1 = \ln(1/(Q T^2))$.  In order to prove the upper and lower bounds of the claim, we first observe that
    \begin{align*}
    \label{eq:understandingdtilde}
      \sum_{x\in S} \widetilde{D}^{t+1}_x&=\frac{1}{(1+2/(Q T^2)) \cdot 2\sqrt{\varepsilon'_t (1 - \varepsilon'_t) }} \sum_{x\in S} \widetilde{D}^t_x \exp\big(-\alpha'_t c(x) h_t(x) + \kappa_{[h_t(x) \neq c(x)]}\big)
      \\
      &= \frac{ (1-\widetilde{\varepsilon}_t)e^{-\alpha'_t + \kappa_0} + \widetilde{\varepsilon}_t e^{\alpha'_t + \kappa_1 } }{2(1+2/(QT^2)) \sqrt{\varepsilon'_t (1 - \varepsilon'_t) }} \tag{using Eq.~\eqref{eq:defnofZtrepeat_b}}
      \\
      &=\frac{ (2 - 1/(Q T^2))(1-\widetilde{\varepsilon}_t)e^{-\alpha'_t}+(1/(Q T^2))\widetilde{\varepsilon}_t e^{\alpha'_t} }{2(1+2/(QT^2)) \sqrt{\varepsilon'_t (1 - \varepsilon'_t) }}     \\
      &=\frac{1}{(1+2/(QT^2) )} \Bigg( \Bigg( 1 - \frac{1}{2 Q T^2} \Bigg)\cdot  \frac{1-\widetilde{\varepsilon}_t}{1-\varepsilon'_t}+  \frac{1}{2 Q T^2} \cdot  \frac{\widetilde{\varepsilon}_t}{\varepsilon'_t}\Bigg),    
  \end{align*}
 where the second equality used $\widetilde{\varepsilon}_t = \sum_{x:h_t(x) \neq c(x)} \widetilde{D}^t_x$, third equality follows by the definition of $\kappa_0,\kappa_1$ and the final equality used $\alpha'_t=\frac{1}{2}\ln \Big(\frac{1-\varepsilon'_t}{\varepsilon'_t} \Big)$.  We now prove the lower bound in the claim: 
  \begin{align*}
            \sum_{x\in S} \widetilde{D}^{t+1}_x&=  \frac{1}{(1+2/(QT^2) )} \Bigg( \Bigg( 1 - \frac{1}{2 Q T^2} \Bigg) \cdot \frac{1-\widetilde{\varepsilon}_t}{1-\varepsilon'_t}+  \frac{1}{2 Q T^2} \cdot \frac{\widetilde{\varepsilon}_t}{\varepsilon'_t}\Bigg) 
            \\
            &\geq \frac{1}{(1+2/(QT^2) )} \Bigg( \Bigg( 1 - \frac{1}{2 Q T^2} \Bigg)\cdot  \frac{1-\widetilde{\varepsilon}_t}{1-\varepsilon'_t} \Bigg)\tag{using $\widetilde{\varepsilon}_t> 0$}\\
      &\geq \frac{1}{(1+2/(QT^2) )} \cdot \Bigg( 1 - \frac{1}{2 Q T^2} \Bigg) \tag{using $1-\widetilde{\varepsilon}_t \geq 1 - \varepsilon'_t$}\\
       &\geq 1 - \frac{3}{ Q T^2} = 1 - 30 \delta \tag{since $\delta = \frac{1}{10 Q T ^2}$ }, 
  \end{align*}
  where we used $\widetilde{\varepsilon}_t \leq \varepsilon'_t$ in the penultimate inequality because we are in the `no' instance of Lemma~\ref{lem:updatingapproxdistribution} in Case II of our proof. We finally get the desired upper bound in the claim as follows
\begin{align*}
\sum_{x\in S} \widetilde{D}^{t+1}_x &=  \frac{1}{(1+2/(QT^2) )} \Bigg( \Bigg( 1 - \frac{1}{2 Q T^2} \Bigg) \cdot \frac{1-\widetilde{\varepsilon}_t}{1-\varepsilon'_t}+  \frac{1}{2 Q T^2} \cdot \frac{\widetilde{\varepsilon}_t}{\varepsilon'_t}\Bigg) 
\\
&\leq   \frac{1}{(1+2/(QT^2) )} \Bigg( \Bigg( 1 - \frac{1}{2 Q T^2} \Bigg) \cdot \frac{1-\widetilde{\varepsilon}_t}{1-1/(QT^2)}+  \frac{1}{2 Q T^2}  \Bigg) \tag{using 
$ \widetilde{\varepsilon}_t \leq \varepsilon'_t = 1/(QT^2)$}
\\
&\leq   \frac{1}{(1+2/(QT^2) )} \Bigg( \Bigg( 1 - \frac{1}{2 Q T^2} \Bigg) \cdot \frac{1}{1-1/(QT^2)}+  \frac{1}{2 Q T^2}  \Bigg) \tag{using $1- \widetilde{\varepsilon}_t \leq 1$}
\\
&\leq  \frac{1}{(1+2/(QT^2) )} \Bigg( \Bigg( 1 - \frac{1}{2 Q T^2} \Bigg) \cdot \Bigg(1 + \frac{2}{QT^2}\Bigg)+  \frac{1}{2 Q T^2}  \Bigg) \leq 1.
\end{align*}
 \end{proof}

\begin{customthm}{4.4}
\label{claim:distancefromtrueeps_a}
Let $t\geq 1$, $\widetilde{\varepsilon}_t = \Pr_{x\sim \widetilde{D}^t} [h_t(x)\neq c(x)]$ be the weighted error corresponding to the  sub-normalized distribution $\widetilde{D}^t$ and $\varepsilon_t = \Pr_{x\sim D^t} [h_t(x)\neq c(x)]$ correspond to the true distribution $D^t$. Then $\vert  \widetilde{\varepsilon}_t - \varepsilon_t \vert \leq 50 \delta$. 
\end{customthm}

\begin{proof} We break down the proof of the claim into two cases.

\textbf{Case I:}  Algorithm~\ref{alg:multiplicativeamplitudeestimation} outputs `yes' in the $t$th iteration. Recall the definition of the sub-normalized distribution $\widetilde{D}^t$ in Eq.~\eqref{eq:QAB_distribution_classicalupdate_a} and the true distribution $D^t$ in Eq.~\eqref{eq:QAB_distribution_classicalupdatetrue}. We then have 
  \begin{align*}
    \vert \widetilde{\varepsilon}_{t+1} - \varepsilon_{t+1} \vert  &= \Bigg\vert \sum_{x\in S} \widetilde{D}^{t+1}_x [h_{t+1}(x)\neq c(x)] - \sum_{x \in S} D^{t+1}_{x} [h_{t+1}(x)\neq c(x)]   \Bigg\vert
    \\
    &\leq \sum_{x\in S} \Big\vert\widetilde{D}^{t+1}_x - D^{t+1}_x   \Big\vert \cdot \Big\vert [h_{t+1}(x)\neq c(x)] \Big\vert 
    \\
    &\leq \sum_{x\in S} \Big\vert \widetilde{D}^{t+1}_x - D^{t+1}_x   \Big\vert  \tag{since $ [h_{t+1}(x)\neq c(x)]  \leq 1$ }
    \\
     &= \sum_{x\in S}  \widetilde{D}^t_x \exp(-\alpha'_t c(x) h_t(x)) \Bigg\vert \frac{1}{ 2(1 + 2 \delta ) \sqrt{\varepsilon'_t (1 - \varepsilon'_t) }} - \frac{1}{(1 - \widetilde{\varepsilon}_t) \cdot  e^{-\alpha'_t} + \widetilde{\varepsilon}_t \cdot  e^{\alpha_t^\prime}}  \Bigg\vert 
     \\
     &= \frac{\sum_{x\in S} \widetilde{D}^t_x \exp(-\alpha'_t c(x) h_t(x))}{(1 - \widetilde{\varepsilon}_t)   e^{-\alpha'_t} + \widetilde{\varepsilon}_t   e^{\alpha'_t}} \Bigg\vert  \frac{(1 - \widetilde{\varepsilon}_t)   e^{-\alpha'_t} + \widetilde{\varepsilon}_t   e^{\alpha'_t}}{2(1 + 2 \delta ) \sqrt{\varepsilon'_t (1 - \varepsilon'_t) }} - 1  \Bigg\vert
     \\
     &= \frac{1}{2 (1 + 2 \delta)} \Bigg\vert \frac{1 - \widetilde{\varepsilon}_t}{1 - \varepsilon'_t} + \frac{\widetilde{\varepsilon}_t}{\varepsilon'_t} - 2(1+2\delta) \Bigg\vert \tag{using Eq.~\eqref{eq:defnofZtrepeat}}
     \\
     &\leq \frac{1}{2 (1 + 2 \delta)} \Bigg( \Bigg\vert \frac{ \widetilde{\varepsilon}_t - \varepsilon'_t}{1 - \varepsilon'_t} \Bigg\vert + \Bigg\vert \frac{ \widetilde{\varepsilon}_t - \varepsilon'_t}{ \varepsilon'_t} \Bigg\vert + 4 \delta \Bigg)  \tag{using triangle inequality }
     \\
     &\leq \frac{\delta}{2 (1 + 2 \delta)} \Bigg( \Bigg\vert \frac{\varepsilon'_t}{1 - \varepsilon'_t} \Bigg\vert + 5 \Bigg) \tag{using $|\widetilde{\varepsilon}_t - \varepsilon'_t|\leq \delta \varepsilon'_t$}
     \\
      &\leq \frac{7 \delta}{2 (1 + 2 \delta)} \leq 4 \delta. \tag{using $ \varepsilon'_t \leq 2/3$}
\end{align*}

  \textbf{Case II:} Algorithm~\ref{alg:multiplicativeamplitudeestimation} outputs `no' in the $t$th iteration. Recall the definition of the sub-normalized distribution $\widetilde{D}^t$ in Eq.~\eqref{eq:QAB_distribution_classicalupdate_b}. Let $\kappa_0 = \ln(2 - 1/(Q T^2))$  and $\kappa_1 = \ln(1/(Q T^2))$. We have 
  \begin{align*}
    &\vert \widetilde{\varepsilon}_{t+1} - \varepsilon_{t+1} \vert\\
    &= \Bigg\vert \sum_{x\in S} \widetilde{D}^{t+1}_x [h_{t+1}(x)\neq c(x)] - \sum_{x \in S} D^{t+1}_{x} [h_{t+1}(x)\neq c(x)]   \Bigg\vert 
     \\
     &\leq   \sum_{x\in S} |\widetilde{D}^{t+1}_x -  D^{t+1}_{x} |\\ 
     &= \sum_{x\in S} \widetilde{D}^t_x \exp(-\alpha'_t c(x) h_t(x) + \kappa_{[h_t(x) \neq c(x) ]} )\cdot    \Bigg\vert  \frac{1}{2(1 + 2/(QT^2) ) \sqrt{\varepsilon'_t (1 - \varepsilon'_t) }} - \frac{1}{(1 - \widetilde{\varepsilon}_t)   e^{-\alpha'_t+\kappa_0}  + \widetilde{\varepsilon}_t   e^{\alpha'_t+\kappa_1} }  \Bigg\vert
     \\
     &= \frac{\sum_{x\in S} \widetilde{D}^t_x \exp(-\alpha'_t c(x) h_t(x) + \kappa_{[h_t(x) \neq c(x) ]} )}{(1 - \widetilde{\varepsilon}_t)   e^{-\alpha'_t+\kappa_0}  + \widetilde{\varepsilon}_t   e^{\alpha'_t+\kappa_1}} \cdot  \Bigg\vert  \frac{(1 - \widetilde{\varepsilon}_t)   e^{-\alpha'_t+\kappa_0}  + \widetilde{\varepsilon}_t   e^{\alpha'_t+ \kappa_1}  }{2(1 + 2/(QT^2) ) \sqrt{\varepsilon'_t (1 - \varepsilon'_t) }} - 1  \Bigg\vert
     \\
     &= \frac{1}{2 (1 + 2/(QT^2 ) )} \Bigg\vert \Bigg(2 - \frac{1}{QT^2} \Bigg) \cdot \frac{1 - \widetilde{\varepsilon}_t}{1 - \varepsilon'_t} + \frac{1}{QT^2} \cdot \frac{\widetilde{\varepsilon}_t}{\varepsilon'_t} - 2\Bigg(1+ \frac{2}{QT^2}\Bigg) \Bigg\vert\tag{using Eq.~\eqref{eq:defnofZtrepeat_b}}
     \\
     &\leq \frac{1}{2 (1 + 2(Q T^2) )} \Bigg( 2 \cdot \Bigg\vert \frac{ 1-\widetilde{\varepsilon}_t }{1 - \varepsilon'_t}-1 \Bigg\vert  + \frac{1}{QT^2} \cdot \Bigg( \Bigg\vert \frac{ 1-\widetilde{\varepsilon}_t }{1 - \varepsilon'_t} \Bigg\vert + \Bigg\vert \frac{ \widetilde{\varepsilon}_t }{ \varepsilon'_t} \Bigg\vert \Bigg) + \frac{4}{QT^2} \Bigg)  \tag{using triangle inequality }
     \\
     &\leq \frac{1}{2 (1 + 2/(QT^2) )} \Bigg(  \frac{2}{QT^2} \cdot \Bigg( \frac{1}{1 - 1/(QT^2) } \Bigg) + \frac{1}{QT^2} \cdot \Bigg( \frac{1}{1 - 1/(QT^2) }+1 \Bigg) + \frac{4}{QT^2} \Bigg) 
     \\
      &\leq \frac{5}{QT^2} = 50 \delta, \tag{using $ \delta = 1/(10QT^2)$}
\end{align*}
where the second last inequality used $0 \leq \widetilde{\varepsilon}_t \leq \varepsilon'_t = 1/(QT^2)$ and $|\widetilde{\varepsilon}_t-\varepsilon'_t|\leq 1/(QT^2)$. 
\end{proof}

\begin{customthm}{4.5}
\label{claim:distancebetweenstates}
 Let $t \geq 1$, $\ket{\Phi_6}= \sum_{x\in S} \sqrt{\widetilde{D}^t_x} \ket{x , c(x)} \ket{0} \ket{0}^t+O_{h_{t-1}}\cdots O_{h_1}\cdot \G_{t}^{-1}\ket{\Psi}$ and $\ket{\Phi'_6}=\sum_{x\in S} \sqrt{{D}^t_x} \ket{x , c(x)} \ket{0} \ket{0}^t$ be defined as in Eq.~\eqref{eq:phi6},~\eqref{eq:phi6'} respectively. Then we have $|\langle \Phi_6 \vert \Phi'_6 \rangle| \geq 1-50 \delta$.
\end{customthm}

\begin{proof}
We break down the proof into two cases. 

 \textbf{Case I:} Algorithm~\ref{alg:multiplicativeamplitudeestimation} outputs `yes' in the $t$th iteration. We now lower bound the inner product~between 
$$
    \ket{\Phi_6}= \sum_{x\in S} \sqrt{\widetilde{D}^t_x} \ket{x , c(x)} \ket{0} \ket{0}^t+\underbrace{O_{h_{t-1}}\cdots O_{h_1}\cdot \G_{t}^{-1}\ket{\Psi}}_{:=\ket{\Psi'}}, \qquad     \ket{\Phi'_6}= \sum_{x\in S} \sqrt{{D}^{t}_x} \ket{x , c(x)} \ket{0} \ket{0}^t.
$$
In order to do so, we first lower bound the following quantity 
\begin{align}
\label{eq:distbetweendist_a}
\begin{aligned}
    \sum_{x\in S} \sqrt{\widetilde{D}^{t+1}_x D^{t+1}_x }&=\sum_{x\in S} \sqrt{\frac{\widetilde{D}^t_x\exp(-\alpha'_t c(x)h_t(x))}{2(1+2 \delta) \sqrt{\varepsilon'_t (1 - \varepsilon'_t) } }}\cdot  \sqrt{\frac{\widetilde{D}^t_x\exp(-\alpha'_t c(x)h_t(x))} {(1-\widetilde{\varepsilon}_t)e^{-\alpha'_t}+\widetilde{\varepsilon}_t e^{\alpha'_t}} }
    \\
    &=\Bigg(\frac{1}{2(1+2 \delta) \sqrt{\varepsilon'_t (1 - \varepsilon'_t) }} \cdot  \frac{1}{(1-\widetilde{\varepsilon}_t)e^{-\alpha'_t}+\widetilde{\varepsilon}_t e^{\alpha'_t}}\Bigg)^{1/2} \sum_{x \in S} \widetilde{D}^t_x\exp(-\alpha'_t c(x)h_t(x))
    \\
    &=\Bigg(\frac{(1-\widetilde{\varepsilon}_t)e^{-\alpha'_t}+\widetilde{\varepsilon}_t e^{\alpha'_t}}{2(1+2 \delta ) \sqrt{\varepsilon'_t (1 - \varepsilon'_t) }} \Bigg)^{1/2} \cdot \frac{\sum_{x \in S} \widetilde{D}^t_x\exp(-\alpha'_t c(x)h_t(x))}{(1-\widetilde{\varepsilon}_t)e^{-\alpha'_t}+\widetilde{\varepsilon}_t e^{\alpha'_t}}   
    \\
    &=\Bigg(\frac{1}{2(1+2 \delta)} \Bigg( \frac{1-\widetilde{\varepsilon}_t}{1-\varepsilon'_t}+ \frac{\widetilde{\varepsilon}_t}{\varepsilon'_t}\Bigg) \Bigg)^{1/2} \cdot 1\geq 1-2\delta,
    \end{aligned}
\end{align}
where the first equality used Eq.~\eqref{eq:QAB_distribution_classicalupdate_a} and~\eqref{eq:QAB_distribution_classicalupdatetrue}, the final equality used Eq.~\eqref{eq:defnofZtrepeat},
and the final inequality used Eq.~\eqref{eq:lowerboundeps/eps'_a} and Eq.~\eqref{eq:lowerbound1-eps/1-eps'_a} to conclude
\begin{align*}
\frac{1}{2(1+2 \delta)} \Bigg( \frac{1-\widetilde{\varepsilon}_t}{1-\varepsilon'_t}+ \frac{\widetilde{\varepsilon}_t}{\varepsilon'_t}\Bigg) \geq 1 - 4 \delta.
\end{align*}
 We are now ready to prove the claim
\begin{align*}
    | \langle \Phi_6 \vert \Phi'_6\rangle |=\Big\vert\sum_{x \in S} \sqrt{\widetilde{D}^t_x D^t_x }+\langle \Psi' \vert \Phi_6\rangle \Big\vert &\geq \Bigg\vert \Big|\sum_{x \in S} \sqrt{ \widetilde{D}^t_x D^t_x }\Big|-\Big|\langle \Psi' \vert \Phi_6\rangle\Big| \Bigg\vert \tag{by reverse triangle inequality}\\
    & \geq  \Big|\sum_{x \in S} \sqrt{\widetilde{D}^t_x D^t_x } \Big|-\Big|\langle \Psi' \vert \Phi_6\rangle\Big|\\
&\geq  1-2\delta -\Big|\langle \Psi' \vert \Phi_6\rangle\Big|\tag{using Eq.\eqref{eq:distbetweendist_a}}\\
    &\geq  1-2\delta -\Big  \| \ket{\Psi'}\Big\|=1-2 \delta -\Big  \| \ket{\Psi}\Big\|\geq 1 - 50 \delta   \tag{using Eq.~\eqref{eq:propertyofPsi}}
\end{align*}
where the penultimate inequality used $
\langle \Psi'\vert \Phi_6\rangle \leq \|\ket{\Psi'}\|\leq  30 \delta$  from Eq.~\eqref{eq:propertyofPsi}.

\textbf{Case II:} Algorithm~\ref{alg:multiplicativeamplitudeestimation} outputs `no' in the $t$th iteration.  
Recall that $\kappa_0 = \ln(2 - 1/(Q T^2))$ and $\kappa_1 = \ln(1/(Q T^2))$. Using Eq.~\eqref{eq:QAB_distribution_classicalupdate_b} and $0 \leq \widetilde{\varepsilon}_t \leq \varepsilon'_t = 1/(QT^2)$, we have
\begin{align}
\label{eq:distbetweendist_b}
\begin{aligned}
    &\sum_{x\in S} \sqrt{\widetilde{D}^{t+1}_x D^{t+1}_x }\\
    &=\sum_{x\in S}\sqrt{\frac{\widetilde{D}^t_x\exp(-\alpha'_t c(x)h_t(x)  + \kappa_{[h_t(x) \neq c(x) ]} )}{2(1+2/(Q T^2)) \sqrt{\varepsilon'_t (1 - \varepsilon'_t) }}}\cdot \sqrt{\frac{\widetilde{D}^t_x\exp(-\alpha'_t c(x)h_t(x)  + \kappa_{[h_t(x) \neq c(x) ]} )}{(1-\widetilde{\varepsilon}_t)e^{-\alpha'_t + \kappa_0}+\widetilde{\varepsilon}_t e^{\alpha'_t + \kappa_1 }}}\\ 
    &=\Bigg(\frac{1}{2(1+2/(Q T^2)) \sqrt{\varepsilon'_t (1 - \varepsilon'_t) }} \cdot  \frac{1}{(1-\widetilde{\varepsilon}_t)e^{-\alpha'_t + \kappa_0}+\widetilde{\varepsilon}_t e^{\alpha'_t + \kappa_1 }}\Bigg)^{1/2}  \cdot \sum_{x\in S} \widetilde{D}^t_x\exp(-\alpha'_t c(x)h_t(x)  + \kappa_{[h_t(x) \neq c(x) ]} ) 
    \\
    &=\Bigg(\frac{(1-\widetilde{\varepsilon}_t)e^{-\alpha'_t + \kappa_0}+\widetilde{\varepsilon}_t e^{\alpha'_t + \kappa_1 }}{2(1+2/(Q T^2)) \sqrt{\varepsilon'_t (1 - \varepsilon'_t) }} \Bigg)^{1/2}  \cdot  \frac{\sum_{x \in S} \widetilde{D}^t_x\exp(-\alpha'_t c(x)h_t(x)  + \kappa_{[h_t(x) \neq c(x) ]} ) }{(1-\widetilde{\varepsilon}_t)e^{-\alpha'_t + \kappa_0}+\widetilde{\varepsilon}_t e^{\alpha'_t + \kappa_1 }}
    \\
    &=\Bigg(\frac{(2 - 1/(Q T^2)) \cdot (1-\widetilde{\varepsilon}_t)e^{-\alpha'_t}+(1/(Q T^2)) \cdot \widetilde{\varepsilon}_t e^{\alpha'_t}}{2(1+2/(Q T^2) ) \sqrt{\varepsilon'_t (1 - \varepsilon'_t) }} \Bigg)^{1/2} 
    \\
  &=\Bigg(  \frac{1}{( 1+2/(QT^2)) } \cdot \Bigg( \Bigg( 1 - \frac{1}{2 Q T^2} \Bigg) \cdot \frac{1-\widetilde{\varepsilon}_t}{1-\varepsilon'_t}+ \Bigg( \frac{1}{2 Q T^2} \Bigg) \cdot \frac{\widetilde{\varepsilon}_t}{\varepsilon'_t} \Bigg) \Bigg)^{1/2} \geq 1 - \frac{3}{2 QT^2}.
    \end{aligned}
\end{align}
The first equality used Eq.~\eqref{eq:QAB_distribution_classicalupdate_b} and Eq.~\eqref{eq:defnofZtrepeat_b}, the fourth equality used the modified distribution update in Eq.~\eqref{eq:QAB_distribution_classicalupdate_b} to conclude
\begin{align*}
\sum_{x \in S}\widetilde{D}^t_x\exp\big(-\alpha'_t c(x)h_t(x)  + \kappa_{[h_t(x) \neq c(x) ]} \big)&=\sum_{x:h_t(x)=c(x)}\widetilde{D}^t_x\exp(-\alpha'_t  + \kappa_{1} )+\sum_{x:h_t(x)\neq c(x)}\widetilde{D}^t_x\exp(\alpha'_t  + \kappa_{0} )\\
    &= (1 - \widetilde{\varepsilon}_t) \cdot  e^{-\alpha'_t + \kappa_0} + \widetilde{\varepsilon}_t \cdot  e^{\alpha_t^\prime + \kappa_1}.
 \end{align*}
 and the last inequality used the lower bound on $\sum_{x\in S} \widetilde{D}^{t}_x$ in Claim 4.3 (Case II).  We are now ready to prove the claim
\begin{align*}
    | \langle \Phi_6 \vert \Phi'_6\rangle |
    & \geq  \Big|\sum_{x \in S} \sqrt{\widetilde{D}^{t}_x D^ t_x} \Big|-\Big|\langle \Psi' \vert \Phi_6\rangle\Big|
    \\
   &\geq  1-  \frac{3}{2QT^2} -\Big|\langle \Psi' \vert \Phi_6\rangle\Big|\tag{using Eq.\eqref{eq:distbetweendist_b}}
   \\
    &\geq  1- \frac{3}{2QT^2} -\Big  \| \ket{\Psi'}\Big\| \\
    &\geq 1-\frac{3}{2QT^2} -\frac{3}{QT^2}  \geq 1-\frac{5}{QT^2} =1 - 50 \delta \tag{using Eq.~\eqref{eq:propertyofPsi} and $\delta= 1/(10 Q T^2)$}.
\end{align*}
This concludes the proof of the claim.
\end{proof}

The proofs of these three claims conclude the description of the quantum boosting algorithm. It remains to prove the correctness of the algorithm.

\subsubsection{Proof of correctness }
\label{sec:proofofboosting}

    \paragraph{Output of quantum algorithm.} After $T$ rounds, the algorithm computes $\varepsilon'_1,\ldots,\varepsilon'_T$  and outputs a hypothesis $H$ which is a weighted combination of the weak hypotheses $\{ h_1, \cdots, h_T \}$ obtained after~$T$ rounds of boosting. The final hypotheses $H$ is then given by 
    $$
    H(x) = \sign \Big( \sum_{t=1}^T \alpha^\prime_t h_t(x) \Big),
    $$
    where $\alpha_t^\prime =  \frac{1}{2} \ln \Big( \frac{ 1 - \varepsilon'_t}{\varepsilon'_t} \Big)$ is the weight obtained in the $t$th iteration. 

      \paragraph{Probability of outputting $H$.} 
    We now bound the probability of failure of the quantum boosting algorithm in obtaining the strong hypothesis $H$. The first source of error is due to amplitude amplification in step $(4)$ of the boosting algorithm, which fails with probability $\leq \frac{1}{3 T}$.  The second error is due to the weak quantum learner failing to output a weak hypothesis in step $(5)$, whose probability is $\leq \frac{1}{3 T}$.  The third source of error is in estimating $\widetilde{\varepsilon}_t$ in step $(8)$, the probability of failure in estimating $\widetilde{\varepsilon}_t$ is $\leq 10\delta/T=O(1/(QT^3))$ (since we set $\delta = 1/(10Q T^2)$). By applying a union bound over the $T$ iterations and all three failure events, we ensure that the overall failure probability of outputting $H$ at the end of our quantum boosting algorithm is an arbitrary constant at most $1/3$ (with a constant overhead in the complexity).
    
    It remains to argue that the training error of $H$ is $\leq 1/10$, i.e., $H(x) = c(x)$ for $9/10$ of the $(x,c(x))$s in $S$. To prove this, we analyze the training error of $H$ with respect to the uniform distribution  $\widetilde{D}^1$ as follows.  We break the proof of correctness into two cases and argue separately. In fact in the first case we will argue that $H$ has zero training error and in the second case we will show the training error of $H$ is at most $1/10$.
    
    \textbf{Case I: } Suppose Algorithm~\ref{alg:multiplicativeamplitudeestimation} outputs `yes' for every $t \in [T]$. This case corresponds to the setting where each weighted error $\widetilde{\varepsilon}_t$ is estimated by an $\varepsilon'_t$ such that $| \varepsilon'_t - \widetilde{\varepsilon_t} | \leq \delta \varepsilon'_t $  for every iteration of the quantum boosting algorithm. In this case
    \begin{align}
    \label{eq:QAB_algorithm_classicaldistribution_a}
    \begin{aligned}
       \widetilde{D}^{t+1}(x) &= \frac{\widetilde{D}^{t}(x)}{Z'_{t}}
       \times \begin{cases} 
      e^{-\alpha_t^\prime} & \text{ if } h_{t}(x) = c(x) \\
      e^{\alpha_t^\prime} & \text{ otherwise }
   \end{cases}
    = \frac{\widetilde{D}^t(x) \exp{\big(-c(x) \alpha_t^\prime h_t(x) \big)} }{Z'_t} .
    \end{aligned}
    \end{align}
     where $Z'_t = 2 (1+2 \delta) \sqrt{ \varepsilon_t^\prime (1 - \varepsilon_t^\prime) }$. By definition, we obtain
     \begin{align}
     \label{eq:QAB_algorithm_classicalfinaldistribution_a}
     \begin{aligned}
     \widetilde{D}^{T+1}(x) &= \widetilde{D}^1(x) \cdot \prod_{t=1}^T \frac{ \exp{ \big( -c(x) \alpha_t^\prime h_t(x)  \big) }}{Z'_t} 
     = \frac{{D}^1(x) \exp{\big( - c(x) \cdot  \sum_{t=1}^T \alpha_t^\prime h_t(x) \big)}  }{\Pi_{t=1}^T Z'_t},
     \end{aligned}
     \end{align}
     where the second equality used $\widetilde{D}^1=D^1$ which is the uniform distribution.
We now upper bound the training error under the distribution $D^1$
     \begin{align}
     \begin{aligned}
     \label{eq:QAB_algorithm_trainingerrorbound_a}
     \Pr_{x\sim {D}^1} [H(x) \neq c(x)] &=\Pr_{x\sim {D}^1} \Big[\sign\Big (\sum_{t=1}^T \alpha_t^\prime h_t(x)\Big) \neq  c(x)\Big]\\
     &\leq \Pr_{x\sim {D}^1} \Big[\exp\Big(-\sum_{t=1}^T \alpha_t^\prime h_t(x) \cdot  c(x)\Big)\Big]
     \\
     &= \sum_{i=1}^M {D}^1(x_i) \exp \Big( - c(x_i) \sum_{t=1}^T \alpha_t^\prime h_t(x_i) \Big) = \sum_{i=1}^M \widetilde{D}^{T+1}(x_i) \Pi_{t=1}^T Z'_t \leq \Pi_{t=1}^T Z'_t, 
     \end{aligned}
     \end{align}
     where the first equality used the definition of $H(x) = \sign ( \sum_{t=1}^T \alpha^\prime_t h_t(x) )$, the  first inequality used $[\sign(z)\neq y]\leq e^{-z\cdot y}$ for $z\in \mathbb{R},y\in \pmset{}$,  the final equality used Eq.~\eqref{eq:QAB_algorithm_classicalfinaldistribution_a} and the final inequality used the fact that $\widetilde{D}^{T+1}$ is a sub-normalized distribution ($ \sum_{x \in S} \widetilde{D}^{T+1}_x \leq 1$). We are now in a stage to analyze the training error of $H$ on $D^1$,
     \begin{align*}
     \label{eq:QAB_algorithm_Finaltrainingerrorbound_a}
      \Pr_{x\sim {D}^1} [H(x) \neq c(x)] \leq \prod_{t=1}^T Z'_t &= (1+2 \delta)^T\prod_{t=1}^T 2 \sqrt{ \varepsilon_t^\prime (1 - \varepsilon_t^\prime) } \tag{using Eq.~\eqref{eq:QAB_algorithm_trainingerrorbound_a} and definition of $Z'_t$} 
      \\
     &\leq e^{2\delta T}\prod_{t=1}^T 2 \sqrt{ \frac{\widetilde{\varepsilon}_t}{1 - \delta} \cdot \Big(1 -\frac{ \widetilde{\varepsilon}_t}{1 + \delta}\Big) } \quad \tag{since  $|\widetilde{\varepsilon}_t - \varepsilon'_t | \leq \delta \varepsilon'_t$}   
     \\
     &\leq e^{2\delta T}\prod_{t=1}^T 2 \sqrt{ \widetilde{\varepsilon}_t (1 + 2 \delta) ( 1 - \widetilde{\varepsilon}_t (1 - \delta) )  }
     \\
     &\leq e^{2\delta T}\prod_{t=1}^T 2 \sqrt{ (\varepsilon_t + 4 \delta ) (1 + 2 \delta) ( 1 - ( \varepsilon_t - 4 \delta ) (1 - \delta) )  } \tag{using $ \vert \widetilde{\varepsilon}_t - \varepsilon_t   \vert \leq 4 \delta$}
     \\
     &\leq e^{2\delta T}\prod_{t=1}^T 2 \sqrt{ \varepsilon_t (1 - \varepsilon_t) + 75\delta}
     \\
     &\leq  e^{2\delta T}\prod_{t=1}^T 2 \sqrt{ 1/4 - \gamma_t^2 + 75 \delta } \quad \tag{since  $ \varepsilon_t \leq 1/2 - \gamma_t$}
     \\
    &= e^{2\delta T}\prod_{t=1}^T \sqrt{ 1 - 4 ( \gamma_t^2 - 75 \delta ) }
     \\
     &\leq e^{2\delta T}\prod_{t=1}^T \sqrt{ 1 - 4 ( \gamma^2 - 75 \delta )  } \quad \tag{since  $\gamma \leq \gamma_t$ for all $t$}
     \\
     &\leq \exp{ \bigg(2 \delta T -2 \sum_{t=1}^T \bigg( \gamma^2 - 75 \delta \bigg) \bigg) } \quad \tag{since  $1 + x \leq e^x$ for $x \in \mathbb{R}$}
     \\
     &\leq \exp{\big(- 2 T \gamma^2 + 16/(QT) \big)} \tag{since  $\delta=1/(10QT^2)$},
     \end{align*}
     where we used Claim~\ref{claim:distancefromtrueeps_a} (Case I) in the third inequality to conclude $|\widetilde{\varepsilon}_t-\varepsilon_t|\leq 4\delta$.  

    In order to conclude the proof-of-correctness, note that for   $T=O((\log M)/\gamma^2)$  and for a sufficiently large constant in the $O(\cdot)$, the final upper bound on the expression is  
    $$\Pr_{x\sim {D}^1} [H(x) \neq c(x)]< 1/M.
    $$
    Since ${D}^1$ is the uniform distribution over $S$, i.e., ${D}^1_x=1/M$ for $(x,c(x))\in S$, this implies that  $\Pr_{x\sim {D}^1} [H(x) \neq c(x)]=0$. Hence $H$ has \emph{zero training error}.
    
    \textbf{Case II:} In this case, we assume that Algorithm~\ref{alg:multiplicativeamplitudeestimation} outputs `no' in the first $\ell \in[T]$ rounds of the quantum boosting Algorithm~\ref{alg:quantumAdaBoost}.\footnote{Our analysis also works when Algorithm~\ref{alg:multiplicativeamplitudeestimation} outputs `no' for arbitrary $\ell$ rounds of the quantum boosting algorithm instead of the first $\ell$ rounds.} We additionally assume that $\ell\leq T/\log (2 \sqrt{Q} T) -1$, which is standard in AdaBoost for the following reason: suppose the  weighted errors of each of the first $t \in [\ell]$  hypotheses satisfies $ \varepsilon_t \leq 1/QT^2\ll 1/3$ (which is the `no' instance of Algorithm~\ref{alg:multiplicativeamplitudeestimation}), then observe that the resulting learner is \emph{strong} and we need not do boosting in the first place. Moreover, suppose $\ell\geq T/\log (2 \sqrt{Q} T)$, then observe that the final hypothesis 
    after the $T$ rounds of AdaBoost has training error at most $1/10$ and we are again done:
    \begin{align*}
     \Pr_{x\sim {D}^1} [H(x) \neq c(x)] =  \Pr_{x\sim {D}^1} \Big[\sign \Big(\sum_{t=1}^{T} \alpha_t h_t(x) \Big) \neq c(x)\Big]  \leq \prod_{t=1}^{T} Z_t &= \prod_{t=1}^{T} 2 \sqrt{ \varepsilon_t (1 - \varepsilon_t) } \\
     &\leq \prod_{t=1}^{\ell} 2 \sqrt{ \varepsilon_t } \leq \Bigg( \frac{2}{\sqrt{Q} T} \Bigg)^{\ell} \leq \frac{1}{10} ,
    \end{align*}
    where the last equality used $\ell\geq T/\log (2 \sqrt{Q} T)$ and $T\geq \log M$. 
    
    So from here onwards we will assume $\ell\leq T/\log (2 \sqrt{Q} T)$ and still show that the training error is at most $1/10$. 
    Note that for the first $\ell$ iterations, the distribution follows the update rule, which defers from the standard AdaBoost update: for every $k\in [\ell]$, 
    \begin{align}
    \label{eq:QAB_algorithm_classicaldistribution_b}
    \begin{aligned}
        \widetilde{D}^{k + 1}_x &= \frac{ \widetilde{D}^{k}_x}{2 (1 + 2/(QT^2))  \sqrt{\varepsilon'_{k} (1 - \varepsilon'_{k} ) }}
          \times \begin{cases} 
          (2 - 1/(QT^2)) e^{-\alpha'_{k}} & \text{ if } h_{k}(x) = c(x) \\
          (1/(QT^2)) e^{\alpha'_{k}} & \text{ otherwise }
          \end{cases}
    \\
    &= \frac{\widetilde{D}^{k}(x) \exp{\big(-\alpha_{k}' \cdot c(x)  h_{k}(x) + \kappa_{[h_k(x) \neq c(x)]} \big) } }{Z'_{k}},
    \end{aligned}
    \end{align}
     where $\varepsilon'_{k} = 1/(QT^2)$ and $Z'_{k} = 2 (1+2/(QT^2) ) \sqrt{ \varepsilon_{k}^\prime (1 - \varepsilon_{k}^\prime) }$. Let $\kappa_0 = \ln(2 - 1/(Q T^2))$  and $\kappa_1 = \ln(1/(Q T^2))$. In particular, observe that for every $\ell\geq 1$, we have
     \begin{align}
\label{eq:recursionforDtilde}
 \widetilde{D}^{\ell+ 1}_x=\frac{ D^1_x}{\Pi_{i=1}^{\ell} Z'_i}\cdot \exp\Big(- c(x)\cdot \sum_{i=1}^{\ell}\alpha'_ih_i(x)\Big)\cdot \exp\Big(\sum_{i=1}^{\ell} \kappa_{[h_i(x)\neq c(x)]} \Big).
\end{align}
     We bound the training error as follows:
     \begin{align*}
     \Pr_{x\sim {D}^1} [H(x) \neq c(x)] 
     &\leq \sum_{x\in S} {D}^1_x \exp \Big( - c(x) \sum_{t=1}^{T} \alpha_t^\prime h_t(x) \Big) 
     \\
     &= \sum_{x\in S} {D}^1_x \exp \Big( - c(x) \sum_{t=1}^{\ell} \alpha_t^\prime h_t(x) \Big) \cdot \exp \Big( - c(x) \sum_{t=\ell+1}^{T} \alpha_t^\prime h_t(x) \Big) 
     \\
     &= \prod_{t=1}^{\ell} Z'_t \sum_{x\in S} \widetilde{D}^{\ell+1}_x \exp \Big( - c(x) \sum_{t=\ell+1}^{T} \alpha_t^\prime h_t(x) \Big) \cdot \exp\Big(-\sum_{i=1}^{\ell} \kappa_{[h_i(x)\neq c(x)]} \Big)
     \\
     &= \prod_{t=1}^T Z'_t  \sum_{x\in S} \widetilde{D}^{T+1}_x \exp\Big(-\sum_{i=1}^{\ell} \kappa_{[h_i(x)\neq c(x)]} \Big)
     \leq  \prod_{t=1}^T Z'_t \cdot (QT^2)^{\ell}\cdot \sum_{x\in S} \widetilde{D}^{T+1}_x\leq \prod_{t=1}^T Z'_t \cdot(QT^2)^{\ell},
     \end{align*}
     where the second equality uses Eq.~\eqref{eq:recursionforDtilde} (i.e., distribution update for `no' instances) and third equality uses Eq.~\eqref{eq:QAB_algorithm_classicalfinaldistribution_a} (i.e., distribution update for the `yes' instances), the penultimate inequality uses $\exp(-\kappa_0)\leq \exp(-\kappa_1)\leq QT^2$ (we remark that this bound is very loose, since $\exp(\kappa_0)=O(1)$) and the final inequality uses the fact that $\widetilde{D}$ is a sub-normalized distribution by Claim~\ref{claim:propertyofD}.
        Continuing to upper bound the above expression, we get
     \begin{align*}
      \Pr_{x\sim {D}^1} [H(x) \neq c(x)] &\leq  (QT^2)^{\ell} \prod_{t=1}^{T} Z'_t 
      \\
      &= \Bigg( (QT^2)^{\ell} \prod_{t=1}^\ell Z'_t \Bigg) \cdot \Bigg( \prod_{t = \ell + 1}^T Z'_t \Bigg)
      \\
      &= \Bigg( (QT^2)^{\ell}  (1 + 2/(QT^2))^{\ell}  \prod_{t=1}^\ell 2\sqrt{ 1/(QT^2) \cdot (1 -  1/(QT^2) )} \Bigg) \cdot \Bigg(  (1+2 \delta)^{T-\ell}  \prod_{t=\ell +1 }^T 2 \sqrt{ \varepsilon_t^\prime (1 - \varepsilon_t^\prime) } \Bigg)  
      \\
      &\leq \Bigg( (QT^2)^{\ell}  (1 + 2/(QT^2))^{\ell}  \prod_{t=1}^\ell 2/\sqrt{QT^2}  \Bigg) \cdot \Bigg( (1+2 \delta)^{T-\ell}  \prod_{t=\ell +1 }^T 2 \sqrt{ \varepsilon_t^\prime (1 - \varepsilon_t^\prime) } \Bigg) \\ 
      &\leq \Bigg( (2 \sqrt{Q} T)^{\ell} \exp{\Big((2 \ell)/(Q T^2) \Big) } - 2 (T-\ell) \gamma^2 + (16T - 16\ell)/(QT^2)  \Bigg)
      \\
      &\leq (2 \sqrt{Q} T)^{\ell}  \exp{\Big(- 2 (T-\ell) \gamma^2 + 16/(QT)\Big)} \leq  \exp{\Big(2\ell (\ln(2\sqrt{Q}T)+\gamma^2)- 2 T \gamma^2 +1 \Big)}, 
      \end{align*}
      where the second equality used
      \[ 
      Z'_t=\begin{cases} 
      2 (1 + 2/(QT^2))  \sqrt{1/(QT^2)\cdot  (1 - 1/(QT^2) ) } & \text{ for } t\leq \ell\\
      2 (1+2 \delta)\cdot   \sqrt{ \varepsilon_t^\prime (1 - \varepsilon_t^\prime) }  & \text{ for } t\geq \ell+1,
      
   \end{cases}
\]
      the third inequality used $1 + x \leq e^x$ for $x \in \mathbb{R}$ and the second factor 
      $$\exp{\Big(- 2 (T-\ell) \gamma^2 + (16T - 16\ell)/(QT^2) \Big)},
      $$ 
      came from the upper bound of the training error derived in Case I  with $T$ replaced by $T - \ell$ (recall that we had showed $\Pr_{x\sim {D}^1} [H(x) \neq c(x)] \leq  \prod_{t=1}^T Z'_t \leq \exp{\big(- 2 T \gamma^2 + 16/(QT) \big)}$). Finally, using $\ell\leq T/\ln (2 \sqrt{Q} T)~-1$, we have
      \begin{align*}
      \Pr_{x\sim {D}^1} [H(x) \neq c(x)]&\leq \exp{\Big(2\ell (\ln(2 \sqrt{Q}T)+\gamma^2)- 2 T \gamma^2  + 1 \Big)}\\
      &\leq \exp{\Big(\Big(\frac{2T}{\ln (2 \sqrt{Q} T)} -2\Big)\cdot \big(\ln(2 \sqrt{Q}T)+\gamma^2\big)- 2 T \gamma^2  + 1 \Big)}\\
      &=\exp\Big(2T-2\gamma^2-2\ln(2\sqrt{Q}T)+\frac{2T\gamma^2}{\ln (2 \sqrt{Q} T)}-2T\gamma^2+1\Big)\leq \frac{e}{4QT^2}\leq \frac{1}{10},
      \end{align*}
      where the final inequality used that $Q,T=O(\log M)$ are sufficiently large. Hence, we have shown that $H$ has training error at most $1/10$.

\subsubsection{Complexity of the algorithm}
    \label{sec:compofboosting}
    First we analyze the \emph{query} complexity of the quantum boosting algorithm (where the query complexity refers to the total number of queries made to the hypothesis-oracles $\{O_{h_1},\ldots,O_{h_T}\}$).  We consider the complexity of the $t$th iteration: in \emph{phase~1}, the number of queries made to $\{h_1, \ldots, h_{t-1} \}$ in order for the weak quantum learner $\A$ to output the hypothesis $h_t$ is at most $\sqrt{M} Q \cdot t$: the~$\sqrt{M}$-factor comes from amplitude amplification and the application of the unitary $\widetilde{W}_t : \ket{0} \rightarrow \ket{\Phi_4}$ involves $Q(t-1)$ queries for the $Q$ copies of the input to the weak learner. An additional $Q (t -~1)$ queries to $\{h_1, \ldots, h_{t-1} \}$ are required while applying $O_{h_1},\ldots,O_{h_{t-1}}$  to uncompute the queries. In \emph{phase~2}, the number of queries made during multiplicative amplitude estimation in order to compute $\varepsilon'_t$ is $\sqrt{M}Q^{3/2}T^3\cdot t$: the $ \sqrt{M}Q^{3/2}T^3$-factor is due to multiplicative amplitude estimation (in Lemma~\ref{lem:updatingapproxdistribution}). Furthermore, each application of $\widetilde{F}_t: \ket{0} \rightarrow \ket{\psi_6}$ involves making $t$ queries. Putting together the contribution from both phases, the total query complexity of the quantum boosting algorithm is  
   $$
   \sum_{t=1}^T \sqrt{M} Q (t-1) + Q (t - 1)+ \sqrt{M}Q^{3/2}T^3 t   = O( \sqrt{M}Q^{3/2}T^5 + \sqrt{M} Q T^2 )=\widetilde{O}( \sqrt{M} Q^{3/2} T^5).
   $$
   
   We now discuss the time complexity of the quantum boosting algorithm. We begin by analyzing the time complexity of the $t$th iteration. Assuming that a quantum RAM can prepare a \emph{uniform} superposition $ \frac{1}{\sqrt{M}} \sum_{x \in S} \ket{x, c(x)} $ using $O(n\log M)$ gates, the time complexity of preparing the initial state $\ket{\psi_1} \otimes \ket{\Phi_1}^{\otimes Q} $ is $O(n Q)$.\footnote{As we mentioned earlier, we could also assume that a quantum learner has access to the uniform quantum examples $\frac{1}{\sqrt{M}}\sum_{x\in S} \ket{x,c(x)}$, in which case we do not need to assume a quantum RAM.} In the second step, in order to prepare $\ket{\psi_2} \otimes \ket{\Phi_2}^{\otimes Q} $, our quantum algorithm uses $O(Q t)$ quantum queries to $\{h_1, \ldots, h_{t-1}\}$ and this can be performed in time $O(Q t )$. The third step involves updating the registers from $\widetilde{D}^1$ to $\widetilde{D}^t$ which requires $Q+~1$ applications of the control unitary $\G_t$. Since there are $t-1$ control qubits and updating the distribution register is an arithmetic operation, the third step for implementing $O(Q)$ operations of $\mathcal{G}_t$ can be performed in $O(n^2 Q t)$ time.\footnote{A quantum circuit can perform arithmetic operations with the same time complexity as a Boolean circuit \cite{kitaev1995quantum}.}
   
   In phase $1$ of the quantum algorithm, we perform amplitude amplification with the unitary $Y_t^{\otimes Q}$ which makes $O(\sqrt{M} Q)$ calls to $\widetilde{W}_t$ and $\widetilde{W}_t^{-1}$. This takes time $O(n^2 \sqrt{M}Q t)$. Next, in order to uncompute the $t-1$ quantum queries in the $Q$ copies, our algorithm uses $O(n Q t)$ time. The weak learner $\mathcal{A}$ takes as input $Q$ samples and outputs a hypothesis~$h_t$ in time $O(n^2 Q)$. Note that we require the quantum learner to output an oracle for $h_t$ instead of explicitly outputting a circuit for $h_t$. 
   
   In phase $2$, the algorithm initially performs an arithmetic  operation $ \sum_{x \in S} \ket{x}\ket{\widetilde{D}_x^t} \ket{ [h_{t}(x) \neq c(x)] } \rightarrow  \sum_{x \in S} \ket{x}  \ket{\widetilde{D}_x^t  [h_{t}(x) \neq c(x)] } \ket{ [h_{t}(x) \neq c(x)] }$ using $O(n)$ gates. Then a controlled reflection operator $V  : \ket{p} \ket{0} \ket{0} \rightarrow \ket{p}  \Big( \sqrt{1-p} \ket{0} + \sqrt{p} \ket{1} \Big) \ket{ \sin^{-1} \big(\sqrt{p}\big) } $ is applied where the operation $\ket{p} \ket{0} \ket{0} \rightarrow \ket{p} \ket{0} \ket{\sin^{-1} \big(\sqrt{p}\big)}$ is an arithmetic process and uses $O(n)$  gates while the operation $\ket{p} \ket{0} \ket{\sin^{-1} \big(\sqrt{p}\big)} \rightarrow \ket{p}  \Big( \sqrt{1-p} \ket{0} + \sqrt{p} \ket{1} \Big) \ket{ \sin^{-1} \big(\sqrt{p}\big) }$ uses one controlled rotation gate. The next step is phase estimation which involves applying QFT using $O(n \cdot \log n)$ gates. The time required for amplitude estimation in order to compute $\varepsilon'_t$ is $O( \sqrt{M}Q^{3/2}T^3 \cdot  t n^2)$: the $\sqrt{M}Q^{3/2}T^3$ calls are made to the unitaries $\widetilde{F}_t , \widetilde{F}_t^{-1}$ and each application of $\widetilde{F}_t: \ket{0} \rightarrow \ket{\psi_6}$ requires  $O(n^2 \cdot t)$ time.  The overall time complexity of the quantum algorithm is  
   \begin{align*}
    \sum_{t=1}^T O\Big(n^2 \sqrt{M}Q^{3/2}T^3  t +  n^2 \sqrt{M} Q t + n^2 Q t \Big) =\widetilde{O}(n^2 \sqrt{M}  Q^{3/2} T^5).
   \end{align*}

 \subsection{Reducing generalization error}
 \label{sec:generalizationofquantum}
 
 In the previous section, we showed that our quantum boosting algorithm produces a hypothesis $H$ that has training error $1/10$ over the training set $\{(x_i,c(x_i))\}_{i \in [M]}$, where $(x_i,c(x_i))$ was sampled according to the unknown distribution $\mathcal{D}$. Now we consider Stage (2) of the algorithm. Recall that the goal of our quantum boosting algorithm  is to output a hypothesis $H:\01^n\rightarrow \{-1,1\}$ that satisfies
 \begin{align}
 \label{eq:requirementofstrong}
 \Pr_{x\sim \mathcal{D}} [H(x)=c(x)]\geq 1-\eta.
\end{align} 
A priori, it is unclear if the output of the quantum boosting algorithm $H$ satisfies Eq.~\eqref{eq:requirementofstrong}. However, we saw in Theorem~\ref{thm:goingfromtrainingtogeneralization} that as long as $M$, i.e., the number of training examples (given as input to the quantum boosting  algorithm) is \emph{large enough}, then not only does $H$ has \emph{zero} training error, but it also ensures small \emph{generalization error}, i.e., $H$ also satisfies Eq.~\eqref{eq:requirementofstrong}.  In particular, Stage~$(2)$ of classical AdaBoost simply uses Theorem~\ref{thm:goingfromtrainingtogeneralization} to argue that: suppose the training error of~$H$ is $0$, then the \emph{generalization error} of $H$ is at most $\eta$ as long as $M\geq O(\VC(\Cc)/\eta^2)$. Using Theorem~\ref{thm:goingfromtrainingtogeneralization}, we can now prove our main theorem: 
   \begin{theorem} [Complexity of Quantum Boosting algorithm]
   \label{finalcomplexityofquantumAdaBoost}
   Fix $\eta>0$ and $\gamma>0$. Let $\Cc=\cup_{n\geq 1}\Cc_n$ be a concept class  and~$\A$ be a $\gamma$-weak quantum PAC algorithm for $\Cc$ that takes time $Q(\Cc)$. Let $n\geq 1$,  $\mathcal{D}:\01^n\rightarrow [0,1]$ be an unknown distribution,  $c\in \Cc_n$ be the unknown target concept and
   $$
   M= \Bigg\lceil\frac{\VC(\Cc)}{\gamma^2}\cdot \frac{\log(\VC(\Cc)/\gamma^2)}{\eta^2 }\Bigg\rceil.
   $$
 Suppose we run Algorithm~\ref{alg:quantumAdaBoost} for $T\geq ((\log M)\cdot \log(1/\delta))/(2\gamma^2)$ rounds, then with probability $\geq 1-\delta$ (over the randomness of the algorithm), we obtain a hypothesis $H$ that has training error $1/10$ and generalization error  $$
   \Pr_{x\sim \mathcal{D}} [H(x)=c(x)]\geq 1-1/10-\eta.
   $$
 Moreover, the time complexity of the quantum boosting algorithm is 
   \begin{align}
       \label{eq:quantumcomplexityofAdaBoost}
   T_Q=O( n^2 \sqrt{M}Q(\Cc)^{3/2}T^5 )=\widetilde{O}\Bigg(\frac{\sqrt{\VC(\Cc)} }{\eta}\cdot Q(\Cc)^{3/2} \cdot \frac{n^2}{\gamma^{11}} \cdot \polylog(1/\delta)\Bigg).
  \end{align}
\end{theorem}
Picking $\eta=1/10$ we get that $H$ has generalization error at most $1/5$. Recall that the complexity of classical AdaBoost using Theorem~\ref{thm:goingfromtrainingtogeneralization} is
$$
T_C=\widetilde{O}\Bigg(\frac{\VC(\Cc)}{\eta^2} \cdot R(\Cc) \cdot \frac{n}{\gamma^4}\log(1/\delta)\Bigg).
$$
In comparison, $T_Q$ is quadratically better than $T_C$ in terms of the $\VC$ dimension of the concept class $\Cc$ and $1/\eta$. Additionally, we could potentially have $Q(\Cc)\ll R(\Cc)$ since the the quantum time complexity of a weak learner can be much lesser than the classical time complexity of learning a concept class as exhibited by~\cite{servediogortler:equiv} (under computational assumptions). 

\paragraph{Open questions.} We conclude with a few interesting questions:  (i) can we improve the polynomial dependence on $1/\gamma$ in the quantum complexity of  boosting? (ii) can we use the quantum boosting algorithm to improve the complexities various quantum algorithms that use \emph{classical} AdaBoost on top of a weak quantum algorithms? (iii) are there \emph{practically relevant} concept classes which have large $\VC$ dimension  for which our quantum boosting algorithm gives a large quantum~speedup, (iv) could one replace the quantum phase estimation step in our quantum boosting algorithm by variational techniques developed by Peruzzo et al.~\cite{vqeforpe}?, (v) could one prove lower bounds on the complexity of quantum boosting?

\paragraph{Acknowledgements.} 
RM would like to thank Seth Lloyd
for his support and guidance during the course of the project.
RM’s research was supported by Seth Lloyd in the Research
Laboratory of Electronics at MIT. Part of this work was
done when SA was a Postdoc at Center for Theoretical
Physics, MIT (funded by the MIT-IBM Watson AI Lab
under the project Machine Learning in Hilbert space) and
visiting Henry Yuen in University of Toronto. SA was also
supported by the Army Research Laboratory and the Army
Research Office under grant number W911NF-20-1-0014. We thank
Ashley Montanaro and Ronald de Wolf for clarifications
regarding multiplicative phase estimation and Ronald de
Wolf for several comments that significantly improved the
presentation of this paper. We thank Steve Hanneke and
Nishant Mehta for various clarifications regarding classical
boosting algorithms. We thank the anonymous reviewers of
ICML 2020 for helpful suggestions on a prelimnary version
of this paper.

\newcommand{\etalchar}[1]{$^{#1}$}

\end{document}